\documentclass[12pt]{article}
\usepackage{geometry}
\usepackage{setspace}
\usepackage{latexsym,amsthm,amsmath,amssymb,url}
\usepackage{enumitem}
\usepackage[round]{natbib}

\usepackage{graphics}
\usepackage[usenames]{color}

\usepackage{tikzsymbols}

\theoremstyle{definition}
\newtheorem{lemma}{Lemma}

\theoremstyle{definition}
\newtheorem{definition}{Definition}
\theoremstyle{definition}
\newtheorem*{definition*}{Definition}
\theoremstyle{definition}
\newtheorem{example}{Example}
\theoremstyle{definition}
\newtheorem{exampleR}{Reduced Example}

\theoremstyle{definition}
\newtheorem*{theorem*}{Theorem}
\theoremstyle{definition}
\newtheorem{theorem}{Theorem}
\theoremstyle{definition}
\newtheorem{algorithm}{Algorithm}

\newtheorem{question}{Question}
\newtheorem*{question*}{Question}
\newtheorem{questionR}{Reduced Question}

\newtheorem{innercustomlemma}{Lemma}
\newenvironment{customlemma}[1]
  {\renewcommand\theinnercustomlemma{#1}\innercustomlemma}
  {\endinnercustomlemma}

\newtheorem{innercustomthm}{Theorem}
\newenvironment{customthm}[1]
  {\renewcommand\theinnercustomthm{#1}\innercustomthm}
  {\endinnercustomthm}

\newcommand{\2}{\vspace{0.2cm}}

\newcommand{\quota}[1]{{q{(#1)}}}

\newcommand{\set}[1]{\left\{#1\right\}}

\newcommand{\Workers}{W}
\newcommand{\worker}[1]{w_{#1}}
\newcommand{\Firms}{F}
\newcommand{\firm}[1]{f_{#1}}

\newcommand{\match}{\mu}

\newcommand{\reduce}{\mathcal{R}}

\newcommand{\WithColor}[1]{#1} \newcommand{\WithoutColor}[1]{}       
\newcommand{\matchingC}[1]{yellow}

\usepackage[pdftex, citecolor=blue, colorlinks]{hyperref}

\begin{document}

\title{Finding all stable matchings with assignment constraints}
\author{
Gregory Z. Gutin\footnote{Computer Science Department, Royal Holloway University of London.} \hspace{.1in} Philip R.\ Neary\footnote{Economics Department, Royal Holloway University of London.} \hspace{.1in} Anders Yeo\footnote{IMADA, University of Southern Denmark.} $^{,}$\footnote{Department of Mathematics, University of Johannesburg.}
}

\date{\today}

\maketitle

\begin{abstract}
\noindent
In this paper we consider stable matchings subject to assignment constraints.
These are matchings that require certain assigned pairs to be included, insist that some other assigned pairs are not, and, importantly, are stable.
Our main contribution is an algorithm, based on the iterated deletion of unattractive alternatives, that determines if assignment constraints are compatible with stability.
Whenever there is a stable matching that satisfies the assignment constraints, our algorithm outputs all of them (each in polynomial time per solution). 
This provides market designers with (i) a tool to test the feasibility of stable matchings subject to assignment constraints, and (ii) a tool to implement them when feasible.
\end{abstract}

\newpage

\section{Introduction}\label{sec:intro}

In a two-sided market with workers on one side and firms on the other, a stable matching is an assignment of workers and firms from which joint deviations do not pay.
In this paper we study many-to-one two-sided matching markets, and our focus is on {\it stable matchings with assignment constraints}.
These are matchings that require some worker-firm pairs to be included, insist that some other worker-firm pairs are not, and, importantly, are stable.

Consider a labour market with junior doctors on one side and hospitals on the other.
Suppose that there is a market designer who is eager for doctor $A$, whose speciality is cardiology, to be matched with one of the four rural hospitals because each, from the designer's perspective, could use another cardiologist.
The designer is also keen for doctor $B$, whose speciality is paediatrics, not to be employed at the downtown hospital since the paediatrics department there is already well-staffed.
Suppose further that these two requirements are ``all that really matter'' to the designer in that how the remaining market participants are matched is immaterial.
Naturally, the designer would like to know:\ does there exist a labour market equilibrium (i.e., a stable matching) that adheres to these criteria?

Any and all variants of the market designer's question above can be addressed by defining a {\it list of assignment constraints}. 
\begin{definition*}\label{}
Given a many-to-one two-sided market with workers, $\Workers$, on one side and firms, $\Firms$, on the other, a {\it list of assignment constraints} specifies a subset of workers $\hat{\Workers} \subseteq \Workers$ and a subset of firms $\hat{\Firms} \subseteq \Firms$ such that:
\begin{enumerate}[label=(\roman*)]
\item
for each worker $\worker{} \in \hat{\Workers}$, there is either a subset of firms, $\Firms^{in}(\worker{})$, one of which must be $\worker{}$'s employer, or there is a subset of firms, $\Firms^{out}(\worker{})$, at which $\worker{}$ must not be employed, or there are both.

\item
for each firm $\firm{} \in \hat{\Firms}$, there is either a subset of workers, $\Workers^{in}(\firm{})$, such that the positions at $\firm{}$ must be filled by workers from the subset, or there is a subset of workers $\Workers^{out}(\firm{})$, each of whom must not be employed at $\firm{}$, or there are both.
\end{enumerate}
\end{definition*}

The main contribution of this paper is an algorithm that answers the following:

\2

{\sc Constrained Stable Matchings}:\ {\em Given a many-to-one two-sided market and a list of assignment constraints, is there is a stable matching that satisfies the assignment constraints?
If yes, generate all such stable matchings.} 

\2

In a big picture sense, {\sc Constrained Stable Matchings} is about questions of the form:\ “is there a stable matching in which...?”
The designer's issue is that matchings of the sort that they want might not be stable.\footnote{We are agnostic on why a market designer might rank different assignments differently. Moreover, we do not restrict the designer's preferences over assignments in any way; as in the example with junior doctors above, we allow for the designer to have strong views on the fortunes of some market participants while being completely disinterested in the circumstances of others.} 
The reason is immediate:\ it is the preferences of individual market participants, and not those of the designer, that determine which assignments are stable.
And implementing a matching that is not stable -- no matter how desirable to the designer -- seems ill advised as it is likely to unravel.
Given this, it would be useful for the designer to know in advance if and when constraints are compatible with stability.
Our algorithm allows the designer to do precisely this.
Furthermore, whenever a stable matching satisfying the constraints exists, our algorithm outputs all of them.

On the surface, {\sc Constrained Stable Matchings} imposes constraints on individual participants.
However, upon reflection, constraints of the form that we consider are really about what combinations of pairs are consistent with stability.
In fact, we will show that {\sc Constrained Stable Matchings} is equivalent to the following, seemingly far simpler, problem that concerns only pairs.

\2

{\sc Reduced Constrained Stable Matchings}: {\em Given two disjoint collections of worker-firm pairs $C^{in}$ and $C^{out}$, is there is a stable matching that includes all worker-firm pairs in $C^{in}$ and does not include any worker-firm pair in $C^{out}$?
If yes, generate all such stable matchings.}

\2

The {rural hospitals theorem} confirms that if a market participant is not matched in some stable matching then they are never matched in any, thereby answering the equivalent of {\sc Reduced Constrained Stable Matchings} for individual participants.\footnote{\label{fn:rural}The observation that the same market participants appear in every stable matching was made first by \cite{McVitieWilson:1970:BIT} for the original one-to-one environment and is sometimes referred as the ``lone wolf theorem''. The result was extended to richer many-to-one environments in \cite{Roth:1984:JPE} and \cite{GaleSotomayor:1985:AMM,GaleSotomayor:1985:DAM} where it was also noted that those on the side of the ``one'' are matched with the same number of individuals in every stable matching. {\cite{Roth:1984:JPE} was a study of the National Resident Matching Program (the NRMP - the system in the United States of America that matches junior doctors to hospitals), and this finding became known as ``the rural hospitals theorem'' due to the fact that the underserved hospitals (that are always underserved) were typically found in supposedly less desirable (rural) locations. The result was further extended in \cite{Roth:1986:E} where it was shown that any hospital with unfilled positions in some stable matching is, it turns out, assigned the same set of junior doctors in every stable matching.}}
While rightly celebrated, the result demonstrates that a market designer who insists on stability is powerless in influencing which collection of market participants ``are in'' and which collection of market participants ``are out''.
We believe this underscores the importance for our focus on what collections of pairs can/cannot appear together in some stable assignment, because influencing who the participants that ``are in'' will be paired with is the only dimension in which a designer can have any sway.
For this reason, we hope that our results will be of interest to practitioners of market design and not purely of theoretical interest.

While we view our algorithm's applicability to questions of interest as our paper's main contribution, we believe that the algorithm has one marked advantage over existing ones.
That advantage is its simplicity.
Our methods are easy to explain and even admit an intuitive graphical interpretation (see the figures throughout).
This stands in stark contrast to many existing algorithms in the matching literature that use machinery that is not only quite involved but is also difficult to ascribe economic meaning.
In fact, arguably all that is required to understand our procedure is familiarity with the definition of an unstable matching.
However, it should be noted that we interpret the causes of instability in a non-standard way.

The balance of the paper is as follows. 
In the remainder of the introduction we provide a quick overview of our algorithm and we discuss how our work fits within the literature.
Section~\ref{sec:Environment} introduces many-to-one matching problems, converts these problems to an equivalent formulation using directed graphs (all our results are obtained using this formulation), and restates our question of interest.
Section~\ref{sec:IDUA} presents the IDUA algorithm that reduces every matching problem to its normal form.
For the reader's convenience, Section~\ref{sec:Algorithm} is split into two parts. In Section \ref{sec:Detail} we formally state our algorithm, prove its correctness, and confirm its time complexity.
In Section \ref{sec:Overview} we demonstrate how the algorithm operates by laboriously working through the paper's main example, Example~\ref{ex:initial}, of a market that a fictitious market designer has placed constraints on.

\subsection{Overview of the algorithm}

Here we provide a sketch of how our algorithm operates.
We are deliberately loose, focusing primarily on the intuition behind the key steps of the algorithm and how it exploits structural properties of a matching market.
As mentioned above, a complete run of the algorithm is described in detail in Section~\ref{sec:Overview}.

The central tool in our approach is the {\it iterated deletion of unattractive alternatives} (IDUA), a procedure introduced first in \cite{BalinskiRatier:1997:} and used by us in \cite{GutinNeary:2023:GEB} to classify unique stable matchings.
IDUA reduces a matching problem by repeatedly deleting matchings that cannot be stable by singling out pairs that do not mesh with stability.
The issue is how to identify such pairs.

To fix ideas, consider a market that includes the workers $\worker{}$ and $\worker{}'$ and the firms $\firm{}$ and $\firm{}'$, and, to keep things simple, let us assume that each firm has only one position.
Now, consider a matching $\match$ that includes the pairs $(\worker{},\firm{}')$ and $(\worker{}',\firm{})$, and suppose further that both $\worker{}$ and $\firm{}$ would prefer to be paired together over their current assignment.
In this case the pair $(\worker{},\firm{})$ is said to be {blocking} for the matching $\match$, and $\match$ is said to be {unstable} as a result.
A matching is stable if it is not unstable.

The matching $\match$ above is deemed unstable because the pair $(\worker{}, \firm{})$ is, in a sense, ``strong'' relative to it.
The thinking is that since both $\worker{}$ and $\firm{}$ are better off if they leave their current partners for each other, that they will do exactly that.\footnote{Of course, once $\worker{}$ and $\firm{}$ pair up, both $\worker{}'$ and $\firm{}'$ are ``back on the market'' and this can have ramifications since each might be deemed an upgrade by other currently matched participants. Dynamics based on chain reactions of this kind are considered in \cite{RothVate:1990:E} and more recently in \cite{rudov:2024:arXivfragile}.}
But we note that the pair $(\worker{}, \firm{})$ is not a pair in $\match$. 
It is a pair {\it outside} of $\match$.
We propose an alternative explanation for why $\match$ is unstable that comes from examining what is going on {\it inside} the matching.
In this case, consider the pairs $(\worker{},\firm{}')$ and $(\worker{}',\firm{})$.
Each of these pairs contributes to the breakaway couple.
But in much the same way that the resulting blocking pair $(\worker{},\firm{})$ can be viewed as strong relative to $\match$, each of these pairs can be regarded as ``weak'' relative to it.
And it is the weaknesses of these pairs that results in the matching being unstable.
Here is the key point:\ there are conditions under which a pair is ``always weak'' in every matching that contains it.
This occurs when one member of the pair has an outside option that not only trumps the current one but is also guaranteed.
Our notion of unattractive alternatives formalises this.

Suppose that not only does $\worker{}$ prefer $\firm{}$ to $\firm{}'$, but that $\worker{}$ prefers $\firm{}$ over {\it all} other firms.
In this case, every matching containing the pair $(\worker{}', \firm{})$ is unstable.
In fact, any matching in which $\firm{}$ is paired with a worker that it prefers less than $\worker{}$ is unstable. 
To see why observe that $\firm{}$ could always (at least) propose to $\worker{}$ and such a proposal would certainly be accepted.
We refer to pairs like $(\worker{}', \firm{})$ as unattractive alternatives, and we note that the submarket with these pairs deleted must have the same set of stable matchings.
An important observation is as follows.
It is of course possible that some worker in a deleted pair has $\firm{}$ as their top choice.
But since being paired with $\firm{}$ is not feasible for this worker, they have a new top choice firm which provides a new guaranteed outside option for that firm.
Thus, deleting an unattractive pair, which rules out all matchings containing the pair, can generate unattractive alternatives that were not unattractive before and this can allow for further deletions.
In other words: iterate.\footnote{In Section 2.3 of \cite{GutinNeary:2023:GEB}, we argue that the IDUA procedure is the matching market analog of the {\it iterated deletion of dominated strategies} for strategic games. The parallels seem natural to us. A player's dominated strategy cannot be part of any solution because it is weak in that the player has an alternative strategy that is guaranteed to pay better no matter what. That is, a dominated strategy is always weak and therefore every strategy profile containing it can be removed from consideration. While a dominated strategy should never be played, its close-but-distant relative, a {\it dominant strategy}, must be chosen by a rational player since it is always the best response. That is, dominant strategies are strong relative to {\it all} strategy profiles. In a matching market a blocking pair is always strong only if both members of the pair rank each other first. Clearly, in both non-cooperative and cooperative environments, the notion of ``always strong'' is a more demanding requirement than ``always weak''.}

Once the IDUA procedure stops -- and it must -- what remains is the {\it normal form}:\ a (sub)matching market that contains all the essential information of the original market in that the set of stable matchings remains unchanged (see Lemma~\ref{lemma:iduaSameMatchings}).
Not only are smaller mathematical objects typically easier to analyse but, as we discuss in Section~\ref{sec:IDUA}, knowledge of the normal form strengthens the rural hospitals theorem free of charge since, in addition to the well-known insights listed in Footnote~\ref{fn:rural}, the normal form also reveals what worker-firms pairs are not part of any stable matching amongst those workers and firms that always feature.

While running IDUA generates the normal form of a matching market, the same procedure can also be used to explore the set of solutions and when they might accord with assignment constraints.
To see how, suppose that some worker $\worker{}$ and some single-position firm $\firm{}$ view each other as acceptable in that each prefers being paired together over being unmatched.
Suppose further that both participants appear in the normal form meaning that $\worker{}$ is employed and the position at $\firm{}$ is filled in every stable matching.
Finally, suppose that the pair $(\worker{}, \firm{})$ is part of the normal form and that the market designer would like to know if there is an equilibrium in which $\worker{}$ is not employed at $\firm{}$.\footnote{If the pair $(\worker{}, \firm{})$ is not in the normal form then there is no stable matching containing the pair.}
Such a requirement has consequences.
Not only can the pair $(\worker{}, \firm{})$ not be included, but it must also be the case that $(\worker{}, \firm{})$ is not a blocking pair against any candidate outcome.

The above hints at the difficulty.
One cannot simply discard all unwanted pairs since one or more of them may be blocking for a matching that satisfies the constraints.
However, it turns out that certain unwanted pairs can be ignored.
The key insight comes from combining Lemma~\ref{lemma:extremalMatches} and Lemma~\ref{lem:NotIn}.

Lemma~\ref{lemma:extremalMatches} notes that both the collection of worker optimal pairs and the collection of firm optimal pairs in the normal form comprise stable matchings.
In fact these are precisely the stable matchings found by running deferred acceptance.
Lemma~\ref{lem:NotIn} then observes that extreme pairs in the normal form can never be blocking in the normal form because, due to the well-known tension in preferences over stable matchings between the opposing sides of the market, if an extreme pair is optimal for one then it is the worst outcome for the other.
This has the consequence that a stable matching not including an extreme pair in the normal form is also a stable matching in the normal form with that pair removed (modulo there being a stable matching without that pair).
So, if an extreme pair is unwanted, then it can be deleted since the only stable matchings forfeited are those that include the extreme pair. 
Finally, the removal of an extreme pair creates a new extreme pair which affords IDUA the potential to reduce the market further since some participant has a new top choice.

While deleting unwanted extreme pairs and then running IDUA is permitted, there remains the question of what to do when no extreme pair is unwanted.
The answer is to select an extreme pair, any extreme pair will do, and ``duplicate'' the environment. 
In the first replica, the algorithm designates this extreme pair as unwanted and proceeds as described above.
In the other copy of the environment, the extreme pair is forced to be included.
If the extreme pair is given by $(\worker{}, \firm{})$, then all other pairs involving $\worker{}$ or containing $\firm{}$ are classified as unwanted and with these new unwanted pairs we proceed as before.
This repeated splitting generates a binary search tree that correctly returns all stable matchings satisfying the constraints (Theorem~\ref{thm:ALLmain}).

\subsection{Related work}

The literature to which our paper is closest falls under the umbrella term of  ``matchings with constraints'' (\cite{KamadaKojima:2017:AER} is a survey).
Constraints can be problematic as they may prove a barrier to stability.
What should a designer do when stability cannot be assured?
There is a growing literature that addresses this issue and proposes mechanisms that generate assignments deemed adequate in other ways.
We do not take a stance on what desiderata an unstable matching should satisfy for it to be deemed up to par.
We simply note that just because a stable matching satisfying constraints is not guaranteed to exist is not a guarantee that one does not. 
Viewed in this way, our algorithm can be used as a first check on whether a particular matching market with constraints requires that stability be relaxed.

To the best of our knowledge, the first paper to consider external constraints in matching markets is the \cite{Roth:1984:JPE} study of the NRMP referenced in Footnote \ref{fn:rural}.
Amongst other things, \citeauthor{Roth:1984:JPE} considered the consequences of two junior doctors who are a couple in their private life and, for example, desire to work at the same hospital above all else.
Roth showed that satisfying constraints of this type may be incompatible with stability.\footnote{\cite{KlausKlijn:2005:JET} and \cite{KlausKlijn:2009:JET} showed that existence is salvaged by restricting the preference domain. \cite{KojimaPathak:2013:QJE} showed that stable matchings with couples exist provided that there are relatively few couples and preference lists are sufficiently short.}
\cite{EhlersHalafir:2014:JET} and \cite{KamadaKojima:2015:AER} consider school choice problems with diversity constraints, and the Japanese medical residency market that has regional caps, respectively.

In theory, finding all stable matchings subject to any kind of constraints can be answered by brute force.
One begins with an algorithm that finds every stable matching \citep{IrvingLeather:1986:SIAMJoC,Gusfield:1987:SIAM,Dworczak:2021:OR, MartnezMasso:2004:MSS, BonifacioJuarez:2022:MSS} and then one compares each against the constraints.
However, the number of stable matchings can be exponential \citep{Knuth:1996:}, and so, despite the fact that each solution can be found in polynomial time \citep{Gusfield:1987:SIAM}, checking all of them may not be feasible in practice.
Our algorithm spends a polynomial amount of time per solution, and so appears of greatest value when the set of constraints is large in which case there are likely to be few solutions. 
In Appendix~\ref{app:exampleExponential}, we present a market with an exponential number of stable matchings but only four that satisfy the constraints.
We know of no other procedure that could identify these four stable matchings in a reasonable time frame.

While the literature on what a designer might or ought do in the presence of constraints is vast, we have not found much work that directly addresses how to find stable matchings satisfying constraints.
The following are notable exceptions.
\cite{GusfieldIrving:1989:} show how to determine the existence of stable matching with {\it forced pairs} (pairs that must be included).
\cite{DiasFonseca:2003:TCS} extend the query of \cite{GusfieldIrving:1989:} to one that finds all stable matchings with both forced pairs and {\it forbidden pairs} (pairs that must not be included).
The combination of forced pairs and forbidden pairs is similar, albeit less general, to our notion of assignment constraints in many-to-one markets.
In addition to considering regret-minimising stable matchings, \cite{MandalRoy:2021:wp} show how to output the two extreme stable matchings satisfying the constraints of \cite{DiasFonseca:2003:TCS}.
Finally, \cite{Garg2020} introduces families of constraints termed lattice-linear predicates.
An example of such a constraint is that one worker's regret is less than that of another.\footnote{The arXiv version, \cite{Garg2018}, developed a parallel algorithm for generating all stable matchings satisfying lattice-linear predicates. We observe that these constraints are quite different to the kind that we consider.}

In the absence of any constraints our algorithm outputs all stable matchings for the market in question.
Given this, it provides another avenue to compute the median stable matching \citep{TeoSethuraman:1998:MOR,Fleiner:2002:wp,SethuramanTeo:2006:MOR,KlausKlijn:2006:IJGT}), the generalised median stable matching  \citep{Cheng:2010:Algorithmica,ChenEgesdal:2016:GEB}, and quantile stable matchings \citep{ChenEgesdal:2021:Games}.\footnote{\cite{Cheng:2010:Algorithmica} shows that computing these objects is in general a computationally intractable problem ({\sf NP}-hard). However, \citeauthor{Cheng:2010:Algorithmica} also describes some families of instances for which the problem can be solved in polynomial time.}
This is useful as median stable matchings are viewed as ``fair'' in that they provide a compromise of sorts between the two sides of the market.\footnote{\cite{Echenique:2024:arXiv} present experimental evidence that the median stable matching is a likely outcome in decentralised two-sided matching markets.}



\section{Constrained matching problems}\label{sec:Environment}

In Section \ref{sec:constrainedProblem} we define matching problems with assignment constraints.
In Section \ref{subsec:reduce} we show how queries about assignments constraints in markets can be reduced to a seemingly simpler question concerning specific pairs.
In Section \ref{sec:digraphs} we show how to reformulate matching problems with assignment constraints using directed graphs, by extending this representation of matching problems, originally due to \cite{Maffray:1992:JCTB}, to include assignment constraints.

\subsection{Matching problems with assignment constraints}\label{sec:constrainedProblem}

The participants in a many-to-one two-sided matching market consist of two disjoint sets of \textit{workers} and \textit{firms}.\footnote{We emphasise that the terms ``workers'' and ``firms'' are merely placeholders. The matching framework of \cite{GaleShapley:1962:AMM} has been used to to analyse a wide range of environments. See \cite{Roth:2008:IJGT} for a survey.} Let $\Workers$ denote the set of $m$ workers and let $\Firms$ denote the set $n$ of firms. Each firm $\firm{} \in \Firms$ has a \textit{quota}, $\quota{\firm{}}$, that represents the maximum number of workers for which $\firm{}$ has positions. Each worker $\worker{} \in \Workers$ can hold at most one position.

Participants rank those on the other side of the market.
Specifically, each worker $\worker{} \in \Workers$ has a strict preference ordering over some nonempty list of firms in $\Firms$, and each firm $\firm{} \in \Firms$ has a strict preference ordering over some nonempty list of workers in $\Workers$.
If firm $\firm{}$ is on the preference list of worker $\worker{}$, we interpret this as $\worker{}$ would prefer to be paired with $\firm{}$ over being unmatched, and we say that $\firm{}$ is {\it acceptable} to $\worker{}$.
An analogous statement holds for the preference lists of firms.
A market participant that is not acceptable is said to be {\it unacceptable}.
We make the simplifying assumption that $\worker{}$ is acceptable to $\firm{}$ if and only if $\firm{}$ is acceptable to $\worker{}$, and we say that $(\worker{}, \firm{})$ is an {\it acceptable pair}.
Finally, we assume that each firm finds some worker acceptable and that each worker finds at least one firm acceptable.

\begin{definition}\label{def:instance}
An instance, $P$, of the {\it many-to-one matching problem} is a set of $m$ workers, $\Workers$, a set of $n$ firms, $\Firms$, and a collection of $m+n$ preference lists, one for each worker and one for each firm.
\end{definition}

A matching, $\match$, is a subset of acceptable pairs such that each worker $\worker{}$ appears at most once and each $\firm{}$ appears in at most $\quota{\firm{}}$ pairs in $\match$.

The following is the key definition of \cite{GaleShapley:1962:AMM}, modified for the environment of this paper.

\begin{definition}\label{def:blocking}
Given a many-to-one two-sided matching market, $P$, we say that worker $\worker{}$ and firm $\firm{}$ form a \textit{blocking pair} with respect to matching $\match$ in $P$, if both prefer each other over their current assignment in $\match$.
That is, both $\worker{}$ and $\firm{}$ find each other acceptable, $(\worker{}, \firm{})$ is not a matched pair in $\match$, and (i) $\worker{}$ is not in a matched pair in $\match$ and $\firm{}$ is in strictly less than $\quota{\firm{}}$ matched pairs in $\match$, or (ii) $\worker{}$ is not in a matched pair in $\match$, $\firm{}$ is in a matched pair with $\worker{}'$ and yet $\firm{}$ prefers $\worker{}$ to $\worker{}'$, or (iii) $\firm{}$ is in strictly less than $\quota{\firm{}}$ matched pairs in $\match$, and $\worker{}$ is in a matched pair with $\firm{}'$ and yet $\worker{}$ prefers $\firm{}$ to $\firm{}'$, or (iv) $\firm{}$ is in $\quota{\firm{}}$ matched pairs, $\worker{}$ is in a matched pair with $\firm{}'$ and there is a worker $\worker{}'$ in a matched pair with $\firm{}$ and both $\worker{}$ prefers $\firm{}$ to $\firm{}'$ and $\firm{}$ prefers $\worker{}$ to $\worker{}'$.
\end{definition}

The solution concept that \cite{GaleShapley:1962:AMM} proposed is termed {\it stability} and is defined by the absence of a blocking pair.

\begin{definition}
A matching $\match$ in $P$ with no blocking pairs is a \textit{stable matching}.
\end{definition}

\cite{GaleShapley:1962:AMM} proved that every instance of the matching problem possesses at least one stable matching.

As stated in {the introduction}, we envisage a hypothetical market designer who is faced with a particular matching market and has preferences over the set of matchings. 
The market designer's difficulty is that matchings that they desire may not be stable.
This tension can be explored by defining assignment constraints as follows.
(While we already provided Definition~\ref{def:constraints} below in the introduction, we repeat it here so as to make this section self-contained. We do likewise for both {\sc Constrained Stable Matchings} and {\sc Reduced Constrained Stable Matchings}.)

\begin{definition}\label{def:constraints}
For a given many-to-one matching problem (see Definition \ref{def:instance}) a {\it list of assignment constraints} specifies a subset of workers $\hat{\Workers} \subseteq \Workers$ and a subset of firms $\hat{\Firms} \subseteq \Firms$ such that
\begin{enumerate}[label=(\roman*)]
\item
for each worker $\worker{} \in \hat{\Workers}$, there is either a subset of firms, $\Firms^{in}(\worker{})$, one of which must be $\worker{}$'s employer, or there is a subset of firms, $\Firms^{out}(\worker{})$, at which $\worker{}$ must not be employed, or there are both.

\item
for each firm $\firm{} \in \hat{\Firms}$, there is either a subset of workers, $\Workers^{in}(\firm{})$, such that the positions at $\firm{}$ must be filled by workers from the subset, or there is a subset of workers $\Workers^{out}(\firm{})$, each of whom must not be employed at $\firm{}$, or there are both.
\end{enumerate}
\end{definition}

For a given list of assignment constraints, the market designer would like to know whether there is a stable matching consistent with the list.
If yes, the market designer would then like to know all of them.
Therefore, the market designer faces the following problem that is the main motivation of our paper.

\2

{\sc Constrained Stable Matchings}:\ {\em Given a many-to-one two-sided market and a list of assignment constraints, is there is a stable matching that satisfies the assignment constraints?
If yes, generate all such stable matchings.
} 

\2

We now provide an example of a many-to-one matching problem with assignment constraints. This will be our working example throughout.
We emphasise two features of this example to indicate that our framework is as general as possible. First, the market is unbalanced (there are more workers than positions), and second, preference lists are not complete (some pairs are unacceptable).

\begin{example}\label{ex:initial}
Consider a many-to-one two-sided matching problem with six workers, $\Workers = \set{\worker{1}, \worker{2}, \worker{3}, \worker{4}, \worker{5}, \worker{6}}$, and four firms, $\Firms = \set{\firm{1}, \firm{2}, \firm{3}, \firm{4}}$. The firms' quotas are given by $\quota{\firm{1}} = \quota{\firm{2}} = \quota{\firm{3}} = 1$ and $\quota{\firm{4}} = 2$.
The preference lists for each worker and each firm are presented below, with the acceptable participants on the other side of the market listed in order of decreasing preference.
\begin{align*}
&\worker{1}: \, \firm{1} \, , \firm{2} \, , \firm{3}, \,  \firm{4}  \hspace{.4in} & \firm{1}: & \,\, \worker{5} \, , \worker{4} \, , \worker{3} \, , \worker{2} \, , \worker{1} \, , \worker{6}			\\
&\worker{2}: \, \firm{2} \, , \firm{1} \, , \firm{4} \, , \firm{3}  \hspace{.4in} & \firm{2}: & \,\, \worker{3} \, , \worker{5} \, , \worker{4} \, , \worker{1} \, , \worker{2} \, , \worker{6}		\\
&\worker{3}: \, \firm{3} \, , \firm{4} \, , \firm{1} \, , \firm{2}  \hspace{.4in} & \firm{3}: &	\,\, \worker{2} \, , \worker{1} \, , \worker{5} \, , \worker{4} \, , \worker{3} 		\\
&\worker{4}: \, \firm{4} \, , \firm{3} \, , \firm{2} \, , \firm{1}  \hspace{.4in} & \firm{4}: & \,\, \worker{5} \, , \worker{1} \, , \worker{2} \, , \worker{3} \, , \worker{4}  \, , \worker{6}
	\\
&\worker{5}: \, \firm{4} \, , \firm{1} \, , \firm{2} \, , \firm{1}  \hspace{.4in} & &
	\\
&\worker{6}: \, \firm{2} \, , \firm{1} \, , \firm{4} \,  \hspace{.4in} & &
\end{align*}

Now return to the hypothetical market designer. Suppose that, for whatever reason, the market designer wants to know if there is a stable matching in which either worker $\worker{1}$ or worker $\worker{6}$ is employed at firm $\firm{2}$, worker $\worker{4}$ is not employed at firm $\firm{1}$, and worker $\worker{6}$ is not employed at firm $\firm{4}$? In the language of Definition \ref{def:constraints}, the market designer would like an answer to the following question:

\begin{question}\label{q:QuestionGeneral}
Given the stable matching problem defined above, is there a stable matching consistent with the list of assignment constraints:\ $\Workers^{out}(\firm{1}) = \set{\worker{4}}$, $\Workers^{in}(\firm{2}) = \set{\worker{1}, \worker{6}}$, and $\Workers^{out}(\firm{4}) = \set{\worker{6}}$?
\end{question}

\end{example}

In the next subsection we show how to reduce {\sc Constrained Stable Matchings} into a far simpler but equivalent formulation.

\subsection{Reducing {\sc Constrained Stable Matchings}}\label{subsec:reduce}

The procedure for reducing {\sc Constrained Stable Matchings} has two main components.

The first step is to run the {\it Iterated Deletion of Unattractive Alternatives} (IDUA) procedure on the equivalent one-to-one environment without constraints.\footnote{The word ``equivalent'' here means that there is a bijection between the sets of stable matchings for the two environments. We quickly sketch how this reduction is carried out towards the end of this section but we note that it is a standard procedure, e.g., \cite{GaleSotomayor:1985:DAM}.}
The reason for reducing to a one-to-one environment is that the IDUA procedure is defined only for one-to-one environments, and the IDUA procedure is the key tool in our algorithm.
Running IDUA on the one-to-one variant reduces the unconstrained environment to its \textit{normal form} -- {a matching problem that contains the same set of stable matchings, and yet is no larger than the original problem.

Step two involves checking how the constraints stack up against the normal form instead of the market as a whole.
Given that the normal form is a (weakly) smaller object, often the list of constraints can be shrunk since some will be rendered redundant.
On occasion, it will be clear that a list of assignment constraints leave {\sc Constrained Stable Matchings} infeasible.

The normal form of a matching problem is always a balanced market (Lemma~\ref{lemma:lem1}) because workers and positions that do not appear in any stable matching have been discarded by IDUA.
As such, assignment constraints that require that a particular worker be employed at one of a certain number of firms is clearly not satisfiable if that worker is never employed in any stable matching.
Similarly, a constraint requiring that the worker not be employed at some subset of firms is rendered redundant if that worker does not appear in the normal form.

The original market of Example \ref{ex:initial} was unbalanced with more workers than positions and so something has to give.
That is, there must be at least one worker (and possibly some position) that is unmatched in every stable matching.
Indeed, we will see that worker $\worker{6}$ is not employed in any stable matching because $\worker{6}$ does not appear in the normal form.
Since $\worker{6}$ appears in the constraint $\Workers^{in}(\firm{2}) = \set{\worker{1}, \worker{6}}$ relevant to firm $\firm{2}$, the constraint is still valid, albeit we reduce from $\Workers^{in}(\firm{2}) = \set{\worker{1}, \worker{6}}$ to $\Workers^{in}(\firm{2}) = \set{\worker{1}}$.
If instead that the market designer's Question \ref{q:QuestionGeneral} had stipulated that $\Firms^{in}(\worker{6})$ was nonempty, then clearly the assignment constraints cannot be satisfied because $\worker{6}$ is never employed in any stable matching.

Once the constraints in Definition \ref{def:constraints} have been reduced as above, we can reformulate them as constraints on worker-firm pairs as follows.

Constraints given by $\Workers^{out}(\firm{})$ and $\Firms^{out}(\worker{})$ are easily handled as they immediately specify pairs that cannot be included in the desired output.

Suppose that worker $\worker{}$ and a position at firm $\firm{}$ both appear in the normal form, meaning that $\worker{}$ is always employed and the position at $\firm{}$ is filled in every stable matching.
Then a constraint of the form $\Firms^{in}(\worker{}) = \set{\firm{}}$ states that the pair $(\worker{}, \firm{})$ must be included in any desired stable matching.

More care must be taken with ``in'' constraints whenever they contain two or more elements.
To see this suppose that $\worker{}$ and a position at firm $\firm{}'$ and a position at firm $\firm{}''$ are all included in the normal form, and consider the constraint $\Firms^{in}(\worker{}) = \set{\firm{}', \firm{}''}$ that stipulates that $\worker{}$ must be employed either at $\firm{}'$ or $\firm{}''$.
By definition $\worker{}$ cannot be simultaneously employed at both firms, and so we cannot simply force the matching to include the pairs $(\worker{}, \firm{}')$ and $(\worker{}, \firm{}'')$.
However, we can instead ask if there is a stable matching at which $\worker{}$ is not employed at any firm other than $\firm{}'$ or $\firm{}''$.
We can do this because $\worker{}$ is employed in every stable matching, and so insisting that $\worker{}$ is not employed at any firm other than $\firm{}'$ or $\firm{}''$ is equivalent to requiring that $\worker{}$ is employed at either $\firm{}'$ or $\firm{}''$.
The constraint $\Firms^{in}(\worker{}) = \set{\firm{}', \firm{}''}$ is therefore equivalent to $\Firms^{out}(\worker{}) = \set{\firm{} : \firm{} \neq \firm{}', \firm{}''}$.
So whenever an ``in'' constraint contains two or more elements, we convert the constraint into the corresponding ``out'' constraint.

We have then shown how to reduce {\sc Constrained Stable Matchings} to a simpler problem, {\sc Reduced Constrained Stable Matchings}, defined as follows. 

\2

{\sc Reduced Constrained Stable Matchings}: {\em Given two disjoint collections of worker-firm pairs $C^{in}$ and $C^{out}$, is there is a stable matching that includes all worker-firm pairs in $C^{in}$ and does not include any worker-firm pair in $C^{out}$?
If yes, generate all such stable matchings.}

\2

We note that the reduction in constraints above requires that each worker and each firm appears at most once in the collection of pairs $C^{in}$.

Let us now recall how to reduce a many-to-one environment to a one-to-one environment.
First we ``split'' each firm $\firm{}$ with $\quota{\firm{}} \geq 2$ (i.e., those firms with two or more positions) into $\quota{\firm{}}$ identical copies, $\firm{}^1, \dots, \firm{}^{\quota{\firm{}}}$, such that each copy has the same preferences over workers as $\firm{}$.
Next, we amend the preferences of workers such that the $\quota{\firm{}}$ copies of firm $\firm{}$ appear consecutively and in the same order order in the preference list of every worker who views firm $\firm{}$ as acceptable.
Doing this yields a new one-to-one matching problem where $\Workers$ is a set of $m$ workers, and $\Firms$ is a set of $\sum_{i=1}^{n}q(\firm{i})$ firms, $\set{\firm{1}^{(1)}, \dots, \firm{1}^{(\quota{\firm{1}})}, \dots, \firm{n}^{(1)}, \dots, \firm{n}^{(\quota{\firm{n}})}}$, with preferences inherited as described.

We now convert the many-to-one matching problem of Example \ref{ex:initial} into its equivalent one-to-one formulation.
This will be our working example going forward.

\begin{exampleR}\label{exR:initial}
Having split firm $\firm{4}$ into two separate positions, $\firm{4}^{1}$ and $\firm{4}^{2}$, and assumed that all workers strictly prefer $\firm{4}^{1}$ to $\firm{4}^{2}$, the preference lists for each worker and position in the corresponding one-to-one matching problem are as presented below.
As before, participants on the other side of the market are listed in order of decreasing preference.
\begin{align*}
&\worker{1}: \, \firm{1} \, , \firm{2} \, , \firm{3} \, , \firm{4}^{1} \, , \firm{4}^{2}  \hspace{.4in} & \firm{1}: & \,\, \worker{5} \, , \worker{4} \, , \worker{3} \, , \worker{2} \, , \worker{1} \, , \worker{6}			\\
&\worker{2}: \, \firm{2} \, , \firm{1} \, , \firm{4}^{1} \, , \firm{4}^{2} \, , \firm{3}  \hspace{.4in} & \firm{2}: & \,\, \worker{3} \, , \worker{5} \, , \worker{4} \, , \worker{1} \, , \worker{2} \, , \worker{6}		\\
&\worker{3}: \, \firm{3} \, , \firm{4}^{1} \, , \firm{4}^{2} \, , \firm{1} \, , \firm{2}  \hspace{.4in} & \firm{3}: &	\,\, \worker{2} \, , \worker{1} \, , \worker{5} \, , \worker{4} \, , \worker{3} 		\\
&\worker{4}: \, \firm{4}^{1} \, , \firm{4}^{2} \, , \firm{3} \, , \firm{2} \, , \firm{1}  \hspace{.4in} & \firm{4}^{1}: & \,\, \worker{5} \, , \worker{1} \, , \worker{2} \, , \worker{3} \, , \worker{4}  \, , \worker{6}
	\\
&\worker{5}: \, \firm{4}^{1} \, , \firm{4}^{2} \, , \firm{1} \, , \firm{2} \, , \firm{1}  \hspace{.4in} & \firm{4}^{2}: & \,\, \worker{5} \, , \worker{1} \, , \worker{2} \, , \worker{3} \, , \worker{4}  \, , \worker{6}
	\\
&\worker{6}: \, \firm{2} \, , \firm{1} \, , \firm{4}^{1} \, , \firm{4}^{2}  \hspace{.4in} & &
\end{align*}

The market designer's Question~\ref{q:QuestionGeneral} can now be expressed more simply as:
\begin{questionR}\label{q:QuestionR}
Is there a stable matching that includes the pair $(\worker{1}, \firm{2})$ or the pair $(\worker{6}, \firm{2})$ and does not include the pair $(\worker{4}, \firm{1})$ nor the pair $(\worker{6}, \firm{4}^{1})$ nor the pair $(\worker{6}, \firm{4}^{2})$?
\end{questionR}

We have determined that $\worker{6}$ is never employed in any stable matching and that all other workers are always employed and all positions are always filled. Therefore, applying the reductions described above the question can be reformulated as querying whether there is a stable matching in which:
\[
C^{in} = \set{(\worker{1}, \firm{2})} \, \text{ and } \, C^{out}= \set{(\worker{4}, \firm{1})}
\]
\end{exampleR}

As mentioned before, our first port of call will be to reduce a matching problem to its normal form using the IDUA procedure.
Before formally introducing the IDUA procedure, that we postpone until Section \ref{sec:IDUA}, we show how to represent a matching problem with assignment constraints as a directed graph.


\subsection{Matching problems with assignment constraints as digraphs\protect\footnote{The graph-theoretic terminology and notation that we employ is identical to that in \cite{GutinNeary:2023:GEB}. For the reader's convenience this notation is recapped in Appendix \ref{sec:notation}.}}\label{sec:digraphs}

Let $P$ be an instance of the one-to-one two-sided matching market where $\Workers$ denotes the set of $m$ workers and $\Firms$ denotes the set of $n$ firms.
(We assume that the reduction from many-to-one to one-to-one has occurred, and so it is possible that some firms are copies of each other.) Given $P$, we define the associated {\em matching digraph}, $D(P) = (V, A)$, where $V$ is the vertex set and $A$ is the arc set. The vertex set $V$ is defined as $V \subseteq \Workers \times \Firms$, where $v = (\worker{}, \firm{}) \in V$ if and only if $(\worker{}, \firm{})$ is an acceptable pair.
The arc set $A$ is defined as follows.

\begin{center}
$\begin{array}{rcl}
  A_{\Workers} & := & \{(\worker{}, \firm{i}) (\worker{}, \firm{j}) \; | \; \worker{} \text{ prefers } \firm{j} \text{ to } \firm{i} \} \\
  A_{\Firms}  & := &  \{ (\worker{i}, \firm{}) (\worker{j}, \firm{}) \; | \; \firm{} \text{ prefers } \worker{j} \text{ to } \worker{i} \} \\
  A & := & A_{\Workers} \cup A_{\Firms} \\
\end{array}$
\end{center}

We now show how to depict the matching problem with constraints of Reduced Example~\ref{exR:initial} using digraphs.
Since there are six workers and five positions, the matching digraph is a $6 \times 5$ grid with rows indexed by workers ($\worker{1}$ through $\worker{6}$) and columns indexed by positions ($\firm{1}$ through $\firm{4}^{2}$).
Preferences are depicted by arcs, horizontal for workers and vertical for firms, that point from a pair to a preferred one.
The matching digraph is depicted in Figure~\ref{fig:ex1}, where we note that arcs implied by transitivity have been suppressed for readability.

\begin{figure}[h!]
\begin{center}
\tikzstyle{vertexDOT}=[scale=0.23,circle,draw,fill]
\tikzstyle{vertexY}=[circle,draw, top color=gray!10, bottom color=gray!40, minimum size=11pt, scale=0.5, inner sep=0.99pt]

\WithColor{\tikzstyle{vertexIn}=[circle,draw, top color=green!20, bottom color=green!60, minimum size=11pt, scale=0.5, inner sep=0.99pt]}
\WithoutColor{\tikzstyle{vertexIn}=[circle,draw, top color=gray!1, bottom color=gray!1, minimum size=21pt, scale=0.5, inner sep=2.99pt]}

\WithColor{\tikzstyle{vertexOut}=[circle,draw, top color=red!20, bottom color=red!60, minimum size=11pt, scale=0.5, inner sep=0.99pt]}
\WithoutColor{\tikzstyle{vertexOut}=[circle,draw, top color=gray!99, bottom color=gray!99, minimum size=21pt, scale=0.5, inner sep=2.99pt]}

\begin{tikzpicture}[scale=0.8]
\node [scale=0.9] at (1,12.5) {$f_1$};
\node [scale=0.9] at (3,12.5) {$f_2$};
\node [scale=0.9] at (5,12.5) {$f_3$};
\node [scale=0.9] at (7,12.5) {$f_4^1$};
\node [scale=0.9] at (9,12.5) {$f_4^2$};

\node [scale=0.9] at (-1,11) {$w_1$};
\node (x11) at (1,11) [vertexY] {$(1,1)$};
\node (x12) at (3,11) [vertexIn] {$(1,2)$};
\node (x13) at (5,11) [vertexY] {$(1,3)$};
\node (x14) at (7,11) [vertexY] {$(1,4)$};
\node (x15) at (9,11) [vertexY] {$(1,5)$};

\node [scale=0.9] at (-1,9) {$w_2$};
\node (x21) at (1,9) [vertexY] {$(2,1)$};
\node (x22) at (3,9) [vertexY] {$(2,2)$};
\node (x23) at (5,9) [vertexY] {$(2,3)$};
\node (x24) at (7,9) [vertexY] {$(2,4)$};
\node (x25) at (9,9) [vertexY] {$(2,5)$};

\node [scale=0.9] at (-1,7) {$w_3$};
\node (x31) at (1,7) [vertexY] {$(3,1)$};
\node (x32) at (3,7) [vertexY] {$(3,2)$};
\node (x33) at (5,7) [vertexY] {$(3,3)$};
\node (x34) at (7,7) [vertexY] {$(3,4)$};
\node (x35) at (9,7) [vertexY] {$(3,5)$};

\node [scale=0.9] at (-1,5) {$w_4$};
\node (x41) at (1,5) [vertexOut] {$(4,1)$};
\node (x42) at (3,5) [vertexY] {$(4,2)$};
\node (x43) at (5,5) [vertexY] {$(4,3)$};
\node (x44) at (7,5) [vertexY] {$(4,4)$};
\node (x45) at (9,5) [vertexY] {$(4,5)$};

\node [scale=0.9] at (-1,3) {$w_5$};
\node (x51) at (1,3) [vertexY] {$(5,1)$};
\node (x52) at (3,3) [vertexY] {$(5,2)$};
\node (x53) at (5,3) [vertexY] {$(5,3)$};
\node (x54) at (7,3) [vertexY] {$(5,4)$};
\node (x55) at (9,3) [vertexY] {$(5,5)$};

\node [scale=0.9] at (-1,1) {$w_6$};
\node (x61) at (1,1) [vertexY] {$(6,1)$};
\node (x62) at (3,1) [vertexY] {$(6,2)$};
\node (x64) at (7,1) [vertexOut] {$(6,4)$};
\node (x65) at (9,1) [vertexOut] {$(6,5)$};

\draw [->,line width=0.03cm] (x15) -- (x14);
\draw [->,line width=0.03cm] (x14) -- (x13);
\draw [->,line width=0.03cm] (x13) -- (x12);
\draw [->,line width=0.03cm] (x12) -- (x11);

\draw [->,line width=0.03cm] (x23) to [out=20,in=160] (x25);
\draw [->,line width=0.03cm] (x25) -- (x24);
\draw [->,line width=0.03cm] (x24) to [out=200,in=340] (x21);
\draw [->,line width=0.03cm] (x21) -- (x22);

\draw [->,line width=0.03cm] (x32) -- (x31);
\draw [->,line width=0.03cm] (x31) to [out=20,in=160] (x35);
\draw [->,line width=0.03cm] (x35) -- (x34);
\draw [->,line width=0.03cm] (x34) -- (x33);

\draw [->,line width=0.03cm] (x41) -- (x42);
\draw [->,line width=0.03cm] (x42) -- (x43);
\draw [->,line width=0.03cm] (x43) to [out=20,in=160] (x45);
\draw [->,line width=0.03cm] (x45) -- (x44);

\draw [->,line width=0.03cm] (x53) -- (x52);
\draw [->,line width=0.03cm] (x52) -- (x51);
\draw [->,line width=0.03cm] (x51) to [out=20,in=160] (x55);
\draw [->,line width=0.03cm] (x55) -- (x54);

\draw [->,line width=0.03cm] (x65) -- (x64);
\draw [->,line width=0.03cm] (x64) to [out=200,in=340] (x61);
  \draw [->,line width=0.03cm] (x61) -- (x62);

\draw [->,line width=0.03cm] (x61) to [out=70,in=290] (x11);
\draw [->,line width=0.03cm] (x11) -- (x21);
\draw [->,line width=0.03cm] (x21) -- (x31);
\draw [->,line width=0.03cm] (x31) -- (x41);
\draw [->,line width=0.03cm] (x41) -- (x51);

\draw [->,line width=0.03cm] (x62) to [out=70,in=290] (x22);
\draw [->,line width=0.03cm] (x22) -- (x12);
\draw [->,line width=0.03cm] (x12) to [out=250,in=110] (x42);
\draw [->,line width=0.03cm] (x42) -- (x52);
\draw [->,line width=0.03cm] (x52) to [out=70,in=290] (x32);

\draw [->,line width=0.03cm] (x33) -- (x43);
\draw [->,line width=0.03cm] (x43) -- (x53);
\draw [->,line width=0.03cm] (x53) to [out=70,in=290] (x13);
\draw [->,line width=0.03cm] (x13) -- (x23);

\draw [->,line width=0.03cm] (x64) to [out=70,in=290] (x44);
\draw [->,line width=0.03cm] (x44) -- (x34);
\draw [->,line width=0.03cm] (x34) -- (x24);
\draw [->,line width=0.03cm] (x24) -- (x14);
\draw [->,line width=0.03cm] (x14) to [out=250,in=110] (x54);

\draw [->,line width=0.03cm] (x65) to [out=70,in=290] (x45);
\draw [->,line width=0.03cm] (x45) -- (x35);
\draw [->,line width=0.03cm] (x35) -- (x25);
\draw [->,line width=0.03cm] (x25) -- (x15);
\draw [->,line width=0.03cm] (x15) to [out=250,in=110] (x55);

\end{tikzpicture}
\end{center}
\caption{The matching digraph representation of Reduced Example \ref{exR:initial}, where arcs implied by transitivity have been suppressed for readability. Assignment constraints are depicted by colour coding vertices.}

\label{fig:ex1}
\end{figure}

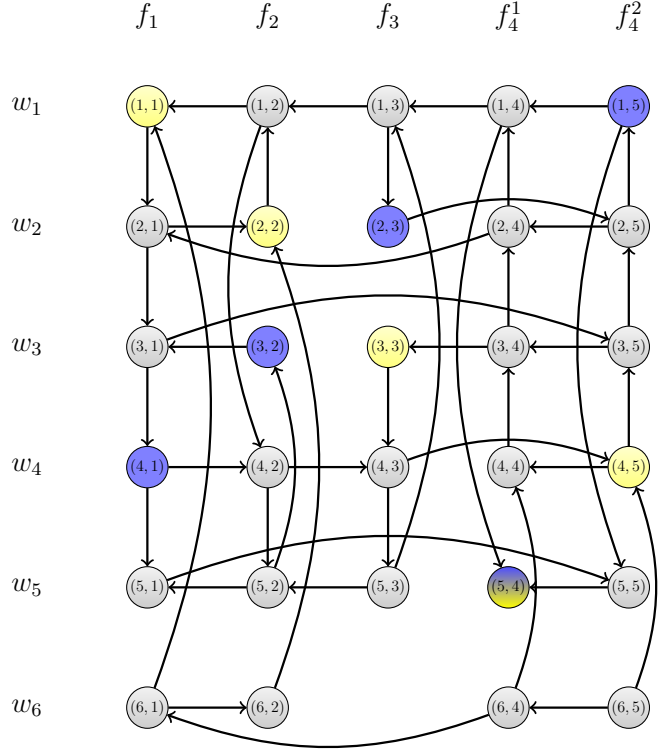
\begin{figure}[h!]
\begin{center}
\tikzstyle{vertexDOT}=[scale=0.23,circle,draw,fill]
\tikzstyle{vertexY}=[circle,draw, top color=gray!10, bottom color=gray!40, minimum size=11pt, scale=0.5, inner sep=0.99pt]

\WithColor{\tikzstyle{vertexZ}=[circle,draw, top color=\matchingC{}!20, bottom color=\matchingC{}!50, minimum size=11pt, scale=0.5, inner sep=0.99pt]} 
\tikzstyle{vertexWF}=[circle,draw, top color=blue!70, bottom color=yellow!100, minimum size=11pt, scale=0.5, inner sep=0.99pt]
\tikzstyle{vertexF}=[circle,draw, top color=blue!50, bottom color=blue!50, minimum size=11pt, scale=0.5, inner sep=0.99pt]
\WithoutColor{\tikzstyle{vertexZ}=[rectangle,draw, top color=gray!1, bottom color=gray!1, minimum size=21pt, scale=0.6, inner sep=2.99pt]}

\tikzstyle{vertexW}=[circle,draw, top color=gray!1, bottom color=gray!1, minimum size=11pt, scale=0.4, inner sep=0.99pt]
\tikzstyle{vertexQ}=[circle,draw, top color=black!99, bottom color=black!99, minimum size=11pt, scale=0.4, inner sep=0.99pt]

\begin{tikzpicture}[scale=0.8]
\node [scale=0.9] at (1,12.5) {$f_1$};
\node [scale=0.9] at (3,12.5) {$f_2$};
\node [scale=0.9] at (5,12.5) {$f_3$};
\node [scale=0.9] at (7,12.5) {$f_4^1$};
\node [scale=0.9] at (9,12.5) {$f_4^2$};

\node [scale=0.9] at (-1,11) {$w_1$};
\node (x11) at (1,11) [vertexZ] {$(1,1)$};
\node (x12) at (3,11) [vertexY] {$(1,2)$};
\node (x13) at (5,11) [vertexY] {$(1,3)$};
\node (x14) at (7,11) [vertexY] {$(1,4)$};
\node (x15) at (9,11) [vertexF] {$(1,5)$};

\node [scale=0.9] at (-1,9) {$w_2$};
\node (x21) at (1,9) [vertexY] {$(2,1)$};
\node (x22) at (3,9) [vertexZ] {$(2,2)$};
\node (x23) at (5,9) [vertexF] {$(2,3)$};
\node (x24) at (7,9) [vertexY] {$(2,4)$};
\node (x25) at (9,9) [vertexY] {$(2,5)$};

\node [scale=0.9] at (-1,7) {$w_3$};
\node (x31) at (1,7) [vertexY] {$(3,1)$};
\node (x32) at (3,7) [vertexF] {$(3,2)$};
\node (x33) at (5,7) [vertexZ] {$(3,3)$};
\node (x34) at (7,7) [vertexY] {$(3,4)$};
\node (x35) at (9,7) [vertexY] {$(3,5)$};

\node [scale=0.9] at (-1,5) {$w_4$};
\node (x41) at (1,5) [vertexF] {$(4,1)$};
\node (x42) at (3,5) [vertexY] {$(4,2)$};
\node (x43) at (5,5) [vertexY] {$(4,3)$};
\node (x44) at (7,5) [vertexY] {$(4,4)$};
\node (x45) at (9,5) [vertexZ] {$(4,5)$};

\node [scale=0.9] at (-1,3) {$w_5$};
\node (x51) at (1,3) [vertexY] {$(5,1)$};
\node (x52) at (3,3) [vertexY] {$(5,2)$};
\node (x53) at (5,3) [vertexY] {$(5,3)$};
\node (x54) at (7,3) [vertexWF] {$(5,4)$};
\node (x55) at (9,3) [vertexY] {$(5,5)$};

\node [scale=0.9] at (-1,1) {$w_6$};
\node (x61) at (1,1) [vertexY] {$(6,1)$};
\node (x62) at (3,1) [vertexY] {$(6,2)$};
\node (x64) at (7,1) [vertexY] {$(6,4)$};
\node (x65) at (9,1) [vertexY] {$(6,5)$};

\draw [->,line width=0.03cm] (x15) -- (x14);
\draw [->,line width=0.03cm] (x14) -- (x13);
\draw [->,line width=0.03cm] (x13) -- (x12);
\draw [->,line width=0.03cm] (x12) -- (x11);

\draw [->,line width=0.03cm] (x23) to [out=20,in=160] (x25);
\draw [->,line width=0.03cm] (x25) -- (x24);
\draw [->,line width=0.03cm] (x24) to [out=200,in=340] (x21);
\draw [->,line width=0.03cm] (x21) -- (x22);

\draw [->,line width=0.03cm] (x32) -- (x31);
\draw [->,line width=0.03cm] (x31) to [out=20,in=160] (x35);
\draw [->,line width=0.03cm] (x35) -- (x34);
\draw [->,line width=0.03cm] (x34) -- (x33);

\draw [->,line width=0.03cm] (x41) -- (x42);
\draw [->,line width=0.03cm] (x42) -- (x43);
\draw [->,line width=0.03cm] (x43) to [out=20,in=160] (x45);
\draw [->,line width=0.03cm] (x45) -- (x44);

\draw [->,line width=0.03cm] (x53) -- (x52);
\draw [->,line width=0.03cm] (x52) -- (x51);
\draw [->,line width=0.03cm] (x51) to [out=20,in=160] (x55);
\draw [->,line width=0.03cm] (x55) -- (x54);

\draw [->,line width=0.03cm] (x65) -- (x64);
\draw [->,line width=0.03cm] (x64) to [out=200,in=340] (x61);
  \draw [->,line width=0.03cm] (x61) -- (x62);

\draw [->,line width=0.03cm] (x61) to [out=70,in=290] (x11);
\draw [->,line width=0.03cm] (x11) -- (x21);
\draw [->,line width=0.03cm] (x21) -- (x31);
\draw [->,line width=0.03cm] (x31) -- (x41);
\draw [->,line width=0.03cm] (x41) -- (x51);

\draw [->,line width=0.03cm] (x62) to [out=70,in=290] (x22);
\draw [->,line width=0.03cm] (x22) -- (x12);
\draw [->,line width=0.03cm] (x12) to [out=250,in=110] (x42);
\draw [->,line width=0.03cm] (x42) -- (x52);
\draw [->,line width=0.03cm] (x52) to [out=70,in=290] (x32);

\draw [->,line width=0.03cm] (x33) -- (x43);
\draw [->,line width=0.03cm] (x43) -- (x53);
\draw [->,line width=0.03cm] (x53) to [out=70,in=290] (x13);
\draw [->,line width=0.03cm] (x13) -- (x23);

\draw [->,line width=0.03cm] (x64) to [out=70,in=290] (x44);
\draw [->,line width=0.03cm] (x44) -- (x34);
\draw [->,line width=0.03cm] (x34) -- (x24);
\draw [->,line width=0.03cm] (x24) -- (x14);
\draw [->,line width=0.03cm] (x14) to [out=250,in=110] (x54);

\draw [->,line width=0.03cm] (x65) to [out=70,in=290] (x45);
\draw [->,line width=0.03cm] (x45) -- (x35);
\draw [->,line width=0.03cm] (x35) -- (x25);
\draw [->,line width=0.03cm] (x25) -- (x15);
\draw [->,line width=0.03cm] (x15) to [out=250,in=110] (x55);

\end{tikzpicture}
\end{center}
\caption{The coloured vertices denote the two extremal stable matchings.}

\label{fig:exExtremeMatchings}
\end{figure}

A {\it matching}, $M$, in $D(P)$ is an independent set of vertices.
Due to the grid-like structure of $D(P)$, independence forbids that any two vertices in $M$ are in the same row or the same column.
A {\it stable matching} in $D(P)$ is a matching such that every vertex in the digraph is either in the matching or has an out-neighbour in the matching. Such an object is known as a {\it kernel}.\footnote{While kernels are by now well-studied objects in graph theory, they were first introduced by \cite{NeumannMorgenstern:1944:} in studying solutions to cooperative games.}

The assignment constraints stipulated by the market designer's {Reduced Question} \ref{q:QuestionR} are depicted by colour coding the vertices in Figure~\ref{fig:ex1}.
The requirement that the position at firm $\firm{2}$ is filled either with worker $\worker{1}$ or worker $\worker{6}$ is conveyed by colouring the vertex $(1, 2)$ in {green}.
This is because worker $\worker{6}$ is not employed at any stable matching, implying that $\worker{1}$ must be employed at $\firm{2}$.
Similarly, the three vertices coloured {red}, $(4, 1), (6, 4)$ and $(6, 5)$, represent those pairs that are forbidden by the market designer, where we note that both $(6,4)$ and $(6,5)$ could be left uncoloured as the entire bottom row (corresponding to $\worker{6}$) is irrelevant since $\worker{6}$ will never be employed in any stable matching.
The remaining vertices shaded grey are those over which the market designer has no opinion.

Figure~\ref{fig:exExtremeMatchings} shows the two extremal stable matchings for this example.
Some vertices are coloured {yellow}, others are coloured {blue}, and vertex $(5,4)$ is a hybrid of yellow and blue. The yellow vertices together with the hybrid vertex make up the worker-optimal stable matching;
the blue vertices together with the hybrid vertex make up the firm-optimal stable matching.
Each of these stable matchings is found by running the deferred acceptance algorithm of \cite{GaleShapley:1962:AMM} with the relevant group in the role of proposer.
We highlight that neither of these stable matchings satisfy the assignment constraints in {Reduced Question} \ref{q:QuestionR} since neither contains the required pair $(\worker{1}, \firm{2})$.
That is, neither contain the {green} vertex of Figure~\ref{fig:ex1}. Finally we note that there are no yellow vertices nor blue vertices in the row corresponding to $\worker{6}$.
This means that $\worker{6}$ does not appear in any stable matching.

Related to the above, we note that the reduced one-to-one two-sided market may be an unbalanced one (e.g., in our example there are more workers than positions).
In the next section we present the IDUA procedure.
One upshot of running IDUA is that it will always convert the unbalanced market into a balanced market such that only those participants that are matched in every stable matching remain.
Moreover, for those participants who are matched in every stable matching, the normal form will reveal some pairings that cannot be part of any stable matching.
This is a strengthening the rural hospitals theorem.


\section{The iterated deletion of unattractive alternatives (IDUA)}\label{sec:IDUA} 

The first step of our algorithm is to run the \textit{Iterated Deletion of Unattractive Alternatives} (IDUA) procedure.
Since our primary focus is on how IDUA assists in addressing matching problems with assignment constraints we will be brief in describing how it arrives at the normal form.
For how the procedure works in general, the interested reader may wish to consult \cite{GutinNeary:2023:GEB}. In Section 2 the procedure and its behavioural interpretation are discussed and in Figure 3 in Section 4.2 an entire run of the IDUA procedure is spelled out.\footnote{In \cite{GutinNeary:2023:GEB}, IDUA is applied to the original one-to-one environment of \cite{GaleShapley:1962:AMM} and is first defined in terms of preferences rather than digraphs. For more general environments, we believe it is easier to apply IDUA by first converting to the digraph formulation.}

IDUA starts out by observing that some pairs in a matching problem are not a barrier to stability.
To see how, suppose that worker $\worker{}$ is the favourite worker of firm $\firm{}$.
Then there cannot be a stable matching in which $\worker{}$ is paired with a firm that $\worker{}$ prefers less than $\firm{}$.
The reason is that $(\worker{}, \firm{})$ would be a blocking pair against any such matching, since $\worker{}$ would propose a joint deviation and $\firm{}$ would certainly accept.
Therefore all firms below $\firm{}$ in $\worker{}$'s preference list can be removed from consideration by $\worker{}$. 
Such firms are {\it unattractive} to $\worker{}$.
The firms that $\worker{}$ deems unattractive reciprocate by removing $\worker{}$ from their preference lists.

The following deletes unattractive pairs using the language of digraphs.

\begin{definition}\label{def:reduce}
Given a matching digraph $D(P)$, we define $\reduce\big(D(P)\big)$ as the subdigraph of $D(P)$ with the following vertices (and any arcs incident on them) deleted.
\begin{enumerate}[label=(\roman*)]
\item\label{item:reduce1}
If no arc in $A_{\Firms}$ leaves $(\worker{}, \firm{}) \in V(D(P))$  then delete all vertices
$(\worker{}, \firm{i})$ such that $(\worker{}, \firm{i}) (\worker{}, \firm{})  \in A_{\Workers}$.
\item\label{item:reduce2}
If no arc in  $A_{\Workers}$ leaves $(\worker{}, \firm{}) \in V(D(P))$  then delete all vertices
$(\worker{i}, \firm{})$ such that $(\worker{i}, \firm{}) (\worker{}, \firm{})  \in A_{\Firms}$.
\end{enumerate}
\end{definition}

The operation described in Definition~\ref{def:reduce} reduces the size of the matching digraph leaving behind a new smaller one.
We note that deleting a vertex does not simply remove a pair, it deletes all matchings that contain that pair so the number of candidate solutions is shrunk.
An important observation is that there may exist vertices that are not removed by applying $\reduce$ to $D(P)$, but which will be removed when applying $\reduce$ to $\reduce\big(D(P)\big)$.
Formally, repeated application of $\reduce$ to $D(P)$ defines the following iterative procedure.

\begin{definition}[The iterated deletion of unattractive alternatives (IDUA)]\label{def:IDUA}
Given an instance of the matching problem, $P$, and its associated matching digraph, $D(P) = \big(V, A_{\Workers} \cup A_{\Firms}\big)$, define $D^{0} = D(P)$, and for each $k \geq 1$, define $D^{k} = \reduce(D^{k-1})$. Finally, define the {\em normal form of} $P$, $D^{*}(P)$, as $D^{k}$ with minimum $k$ such that $D^{k} = \reduce(D^{k}),$ i.e., no further reductions will take place for such $D^{k}.$
\end{definition}

Lemma~2 in \cite{GutinNeary:2023:GEB} shows that $D^{*}(P)$ is uniquely defined.
This is important to point out since clearly there may be multiple vertices that satisfy the requirements (e.g., part \ref{item:reduce1} in Definition~\ref{def:reduce}).
The lemma assures that the order in which vertices satisfying the condition are chosen (and hence the order in which other vertices are deleted) will not impact the final outcome.

The following lemma confirms that while the IDUA procedure can discard information, no ``important'' information is lost in the sense that the set of stable matchings for $D(P)$ can be computed using only the normal form $D^*(P)$.

\begin{lemma}[\cite{BalinskiRatier:1997:} and \cite{GutinNeary:2023:GEB}]\label{lemma:iduaSameMatchings}
The iterated deletion of unattractive alternatives does not change the set of stable matchings.  That is, $D(P)$ and $D^*(P)$ contain exactly the same stable matchings.
\end{lemma}

Let us now show the effect that running IDUA has on the matching problem of Reduced Example~\ref{exR:initial}.
In Figure~\ref{fig:ex2} there are two panels.
The left hand panel, panel (a), depicts $D(P)$ and is the same as that in Figure~\ref{fig:ex1}.
The right hand panel, panel (b), depicts the normal form, $D^*(P)$.
In the interest of space we will not display here how the IDUA procedure incrementally reduced $D(P)$ to $D^*(P)$.

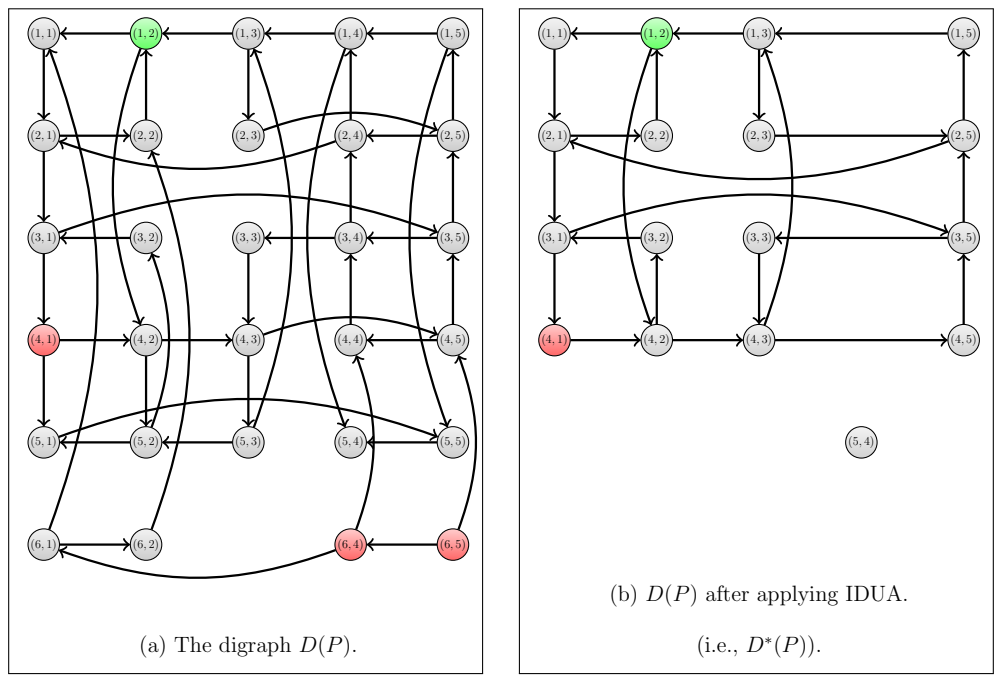
\begin{figure}[h!]
\begin{center}
\fbox{
\begin{minipage}{5.8cm}
\tikzstyle{vertexDOT}=[scale=0.12,circle,draw,fill]
\tikzstyle{vertexY}=[circle,draw, top color=gray!10, bottom color=gray!40, minimum size=11pt, scale=0.38, inner sep=0.99pt]

\WithColor{\tikzstyle{vertexIn}=[circle,draw, top color=green!20, bottom color=green!60, minimum size=11pt, scale=0.38, inner sep=0.99pt]}
\WithoutColor{\tikzstyle{vertexIn}=[circle,draw, top color=gray!1, bottom color=gray!1, minimum size=21pt, scale=0.38, inner sep=2.99pt]}

\WithColor{\tikzstyle{vertexOut}=[circle,draw, top color=red!20, bottom color=red!60, minimum size=11pt, scale=0.38, inner sep=0.99pt]}
\WithoutColor{\tikzstyle{vertexOut}=[circle,draw, top color=gray!99, bottom color=gray!99, minimum size=21pt, scale=0.38, inner sep=2.99pt]}

\begin{tikzpicture}[scale=0.68]

\node (x11) at (1,11) [vertexY] {$(1,1)$};
\node (x12) at (3,11) [vertexIn] {$(1,2)$};
\node (x13) at (5,11) [vertexY] {$(1,3)$};
\node (x14) at (7,11) [vertexY] {$(1,4)$};
\node (x15) at (9,11) [vertexY] {$(1,5)$};

\node (x21) at (1,9) [vertexY] {$(2,1)$};
\node (x22) at (3,9) [vertexY] {$(2,2)$};
\node (x23) at (5,9) [vertexY] {$(2,3)$};
\node (x24) at (7,9) [vertexY] {$(2,4)$};
\node (x25) at (9,9) [vertexY] {$(2,5)$};

\node (x31) at (1,7) [vertexY] {$(3,1)$};
\node (x32) at (3,7) [vertexY] {$(3,2)$};
\node (x33) at (5,7) [vertexY] {$(3,3)$};
\node (x34) at (7,7) [vertexY] {$(3,4)$};
\node (x35) at (9,7) [vertexY] {$(3,5)$};

\node (x41) at (1,5) [vertexOut] {$(4,1)$};
\node (x42) at (3,5) [vertexY] {$(4,2)$};
\node (x43) at (5,5) [vertexY] {$(4,3)$};
\node (x44) at (7,5) [vertexY] {$(4,4)$};
\node (x45) at (9,5) [vertexY] {$(4,5)$};

\node (x51) at (1,3) [vertexY] {$(5,1)$};
\node (x52) at (3,3) [vertexY] {$(5,2)$};
\node (x53) at (5,3) [vertexY] {$(5,3)$};
\node (x54) at (7,3) [vertexY] {$(5,4)$};
\node (x55) at (9,3) [vertexY] {$(5,5)$};

\node (x61) at (1,1) [vertexY] {$(6,1)$};
\node (x62) at (3,1) [vertexY] {$(6,2)$};
\node (x64) at (7,1) [vertexOut] {$(6,4)$};
\node (x65) at (9,1) [vertexOut] {$(6,5)$};

\draw [->,line width=0.03cm] (x15) -- (x14);
\draw [->,line width=0.03cm] (x14) -- (x13);
\draw [->,line width=0.03cm] (x13) -- (x12);
\draw [->,line width=0.03cm] (x12) -- (x11);

\draw [->,line width=0.03cm] (x23) to [out=20,in=160] (x25);
\draw [->,line width=0.03cm] (x25) -- (x24);
\draw [->,line width=0.03cm] (x24) to [out=200,in=340] (x21);
\draw [->,line width=0.03cm] (x21) -- (x22);

\draw [->,line width=0.03cm] (x32) -- (x31);
\draw [->,line width=0.03cm] (x31) to [out=20,in=160] (x35);
\draw [->,line width=0.03cm] (x35) -- (x34);
\draw [->,line width=0.03cm] (x34) -- (x33);

\draw [->,line width=0.03cm] (x41) -- (x42);
\draw [->,line width=0.03cm] (x42) -- (x43);
\draw [->,line width=0.03cm] (x43) to [out=20,in=160] (x45);
\draw [->,line width=0.03cm] (x45) -- (x44);

\draw [->,line width=0.03cm] (x53) -- (x52);
\draw [->,line width=0.03cm] (x52) -- (x51);
\draw [->,line width=0.03cm] (x51) to [out=20,in=160] (x55);
\draw [->,line width=0.03cm] (x55) -- (x54);

\draw [->,line width=0.03cm] (x65) -- (x64);
\draw [->,line width=0.03cm] (x64) to [out=200,in=340] (x61);
  \draw [->,line width=0.03cm] (x61) -- (x62);

\draw [->,line width=0.03cm] (x61) to [out=70,in=290] (x11);
\draw [->,line width=0.03cm] (x11) -- (x21);
\draw [->,line width=0.03cm] (x21) -- (x31);
\draw [->,line width=0.03cm] (x31) -- (x41);
\draw [->,line width=0.03cm] (x41) -- (x51);

\draw [->,line width=0.03cm] (x62) to [out=70,in=290] (x22);
\draw [->,line width=0.03cm] (x22) -- (x12);
\draw [->,line width=0.03cm] (x12) to [out=250,in=110] (x42);
\draw [->,line width=0.03cm] (x42) -- (x52);
\draw [->,line width=0.03cm] (x52) to [out=70,in=290] (x32);

\draw [->,line width=0.03cm] (x33) -- (x43);
\draw [->,line width=0.03cm] (x43) -- (x53);
\draw [->,line width=0.03cm] (x53) to [out=70,in=290] (x13);
\draw [->,line width=0.03cm] (x13) -- (x23);

\draw [->,line width=0.03cm] (x64) to [out=70,in=290] (x44);
\draw [->,line width=0.03cm] (x44) -- (x34);
\draw [->,line width=0.03cm] (x34) -- (x24);
\draw [->,line width=0.03cm] (x24) -- (x14);
\draw [->,line width=0.03cm] (x14) to [out=250,in=110] (x54);

\draw [->,line width=0.03cm] (x65) to [out=70,in=290] (x45);
\draw [->,line width=0.03cm] (x45) -- (x35);
\draw [->,line width=0.03cm] (x35) -- (x25);
\draw [->,line width=0.03cm] (x25) -- (x15);
\draw [->,line width=0.03cm] (x15) to [out=250,in=110] (x55);

\node [scale=0.7] at (5,-1) {(a) The digraph $D(P)$.};
\end{tikzpicture}

\end{minipage} } \hspace{0.2cm} \fbox{
\begin{minipage}{5.8cm}
\tikzstyle{vertexDOT}=[scale=0.12,circle,draw,fill]
\tikzstyle{vertexY}=[circle,draw, top color=gray!10, bottom color=gray!35, minimum size=11pt, scale=0.38, inner sep=0.99pt]

\WithColor{\tikzstyle{vertexIn}=[circle,draw, top color=green!20, bottom color=green!60, minimum size=11pt, scale=0.38, inner sep=0.99pt]}
\WithoutColor{\tikzstyle{vertexIn}=[circle,draw, top color=gray!1, bottom color=gray!1, minimum size=21pt, scale=0.38, inner sep=2.99pt]}

\WithColor{\tikzstyle{vertexOut}=[circle,draw, top color=red!20, bottom color=red!60, minimum size=11pt, scale=0.38, inner sep=0.99pt]}
\WithoutColor{\tikzstyle{vertexOut}=[circle,draw, top color=gray!99, bottom color=gray!99, minimum size=21pt, scale=0.38, inner sep=2.99pt]}

\begin{tikzpicture}[scale=0.68]

\node (x11) at (1,11) [vertexY] {$(1,1)$};
\node (x12) at (3,11) [vertexIn] {$(1,2)$};
\node (x13) at (5,11) [vertexY] {$(1,3)$};
\node (x15) at (9,11) [vertexY] {$(1,5)$};

\node (x21) at (1,9) [vertexY] {$(2,1)$};
\node (x22) at (3,9) [vertexY] {$(2,2)$};
\node (x23) at (5,9) [vertexY] {$(2,3)$};
\node (x25) at (9,9) [vertexY] {$(2,5)$};

\node (x31) at (1,7) [vertexY] {$(3,1)$};
\node (x32) at (3,7) [vertexY] {$(3,2)$};
\node (x33) at (5,7) [vertexY] {$(3,3)$};
\node (x35) at (9,7) [vertexY] {$(3,5)$};

\node (x41) at (1,5) [vertexOut] {$(4,1)$};
\node (x42) at (3,5) [vertexY] {$(4,2)$};
\node (x43) at (5,5) [vertexY] {$(4,3)$};
\node (x45) at (9,5) [vertexY] {$(4,5)$};

\node (x54) at (7,3) [vertexY] {$(5,4)$};


\draw [->,line width=0.03cm] (x15) -- (x13);
\draw [->,line width=0.03cm] (x13) -- (x12);
\draw [->,line width=0.03cm] (x12) -- (x11);

\draw [->,line width=0.03cm] (x23) -- (x25);
\draw [->,line width=0.03cm] (x25) to [out=200,in=340] (x21);
\draw [->,line width=0.03cm] (x21) -- (x22);

\draw [->,line width=0.03cm] (x32) -- (x31);
\draw [->,line width=0.03cm] (x31) to [out=20,in=160] (x35);
\draw [->,line width=0.03cm] (x35) -- (x33);

\draw [->,line width=0.03cm] (x41) -- (x42);
\draw [->,line width=0.03cm] (x42) -- (x43);
\draw [->,line width=0.03cm] (x43) -- (x45);

\draw [->,line width=0.03cm] (x11) -- (x21);
\draw [->,line width=0.03cm] (x21) -- (x31);
\draw [->,line width=0.03cm] (x31) -- (x41);

\draw [->,line width=0.03cm] (x22) -- (x12);
\draw [->,line width=0.03cm] (x12) to [out=250,in=110] (x42);
\draw [->,line width=0.03cm] (x42) -- (x32);

\draw [->,line width=0.03cm] (x33) -- (x43);
\draw [->,line width=0.03cm] (x43) to [out=70,in=290] (x13);
\draw [->,line width=0.03cm] (x13) -- (x23);

\draw [->,line width=0.03cm] (x45) -- (x35);
\draw [->,line width=0.03cm] (x35) -- (x25);
\draw [->,line width=0.03cm] (x25) -- (x15);

\node [scale=0.7] at (5,0) {(b) $D(P)$ after applying IDUA.};
\node [scale=0.7] at (5,-1) {(i.e., $D^*(P)$).};
\end{tikzpicture}
\end{minipage} }
\end{center}
\caption{In (a) we depict $D(P)$ associated with Reduced Example~\ref{exR:initial}.
After applying IDUA we obtain the digraph, $D^*(P)$, depicted in (b).}
\label{fig:ex2}
\end{figure}

We now highlight some observations concerning the structure of $D^*(P)$ in panel (b) of Figure~\ref{fig:ex2}.
First of all, while it is not possible to display all the stable matchings on one image, by comparing the digraph in panel (b) of Figure~\ref{fig:ex2} with the digraph in Figure~\ref{fig:exExtremeMatchings} one can see that the worker-optimal and firm-optimal stable matchings of $D(P)$ remain in $D^*(P)$. (As indeed they must by Lemma~\ref{lemma:iduaSameMatchings}.)
In fact, closer inspection reveals that the worker-optimal and firm-optimal stable matchings of $D(P)$ are made up of the collections of vertices that have no horizontal arcs out of them in $D^*(P)$ and the vertices that have no vertical arcs out of them in $D^*(P)$, respectively (see Lemma~\ref{lemma:extremalMatches}).
Thus, the extremal stable matchings of $D^*(P)$ are precisely those found by running deferred acceptance twice, once with each side in the role of proposer.
Or, to put it another way, the normal form $D^*(P)$ is a smaller matching market than $D(P)$ that is ``boxed in'' by the two extremal stable matchings since all information from preferences lists ``beyond'' these two extremal matchings is removed.

The second thing to note is that there is only one vertex in row 5 and column 4 of $D^*(P)$.
This means that worker $\worker{5}$ is certain to be employed at firm $\firm{4}$ in every stable matching. This same observation could be made by examining Figure~\ref{fig:exExtremeMatchings} wherein vertex $(5,4)$ was a hybrid of the colours yellow and blue indicating that it was in both the worker-optimal stable matching and the firm-optimal stable matching, and hence, due to the lattice structure of the set of stable matchings, must be included in every stable matching.
We also note that there is no vertex in the row labelled by $\worker{6}$ of $D^*(P)$. From this we can conclude that in every stable matching all positions are filled but worker $\worker{6}$ is never employed.

A {\it balanced} two-sided market is one with the same number of workers as positions.
A {\it perfect} matching in a balanced two-sided market is one in which every worker is employed and every position is filled.
The third and final observation we make regarding panel (b) of Figure~\ref{fig:ex2} concerns balanced markets and perfect stable matchings.
While $D^*(P)$ is a balanced market, in this case with five workers and five positions, it is also the case that every stable matching is a perfect matching with all five workers employed and all five positions filled.

The next two results, Lemma~\ref{lemma:extremalMatches} and Lemma~\ref{lemma:lem1}, show that the observations above are not artefacts of this specific example.
Rather they are general properties that hold for the normal form of every two-sided matching market.
Before the formal statements, some additional notation is needed.
Let $d^{*,+}_{\Workers}(x)$ denote the number of arcs out of $x \in V(D^*(P))$
in $A_{\Workers}$ that still remain in $D^*(P)$  and analogously let $d^{*,+}_{\Firms}(x)$ denote the number of arcs out of $x \in V(D^*(P))$ in $A_{\Firms}$ that still remain in $D^*(P)$.

\begin{lemma}[\cite{BalinskiRatier:1997:} and \cite{GutinNeary:2023:GEB}] \label{lemma:extremalMatches}
Let $P$ be an instance of the matching problem. The following two sets of vertices, $M_{\Workers}$ and $M_{\Firms}$, are both stable matchings in $P$.

\begin{center}
$\begin{array}{rcl} \vspace{0.2cm}
 M_{\Workers} & = & \{ x \in V(D^*(P)) \; | \; d_{\Workers}^{*,+}(x)=0 \} \\
 M_{\Firms} & = & \{ x \in V(D^*(P)) \; | \; d_{\Firms}^{*,+}(x)=0 \} \\
\end{array}$
\end{center}
\end{lemma} 

Lemma~\ref{lemma:extremalMatches} states that all remaining participants on the same side of the market have a different most-preferred partner in the normal form, and that such an assignment must be a stable matching.
Clearly such an assignment must be a matching because otherwise two participants on the same side, two workers say, have the same most-preferred firm.
But if this were the case, then the IDUA procedure is not yet finished.
Such a matching must also be a stable matching because every participant on one side of the market is paired with their most preferred feasible partner, implying that no participant on that side of the market is willing to break with their current partner which means there cannot be any blocking pairs.
That is, the extremal matchings in $D^*(P)$ are precisely the extremal stable matchings in $D(P)$.

The next result, the proof of which is found in Appendix \ref{app:proofs}, confirms that an analyst interested in stable outcomes in two-sided matching markets can always confine attention to perfect stable matchings in balanced markets.
In fact, the result extends the rural hospitals theorem as it also observes that some worker-position pairs may be discarded even amongst those workers and positions that appear in every stable matching.
For example, from panel (b) in Figure~\ref{fig:ex2} it is clear that worker $\worker{5}$ will never be paired with firm $\firm{2}$ in any stable matching.
And this despite both appearing in every stable matching.

\begin{lemma} \label{lemma:lem1} 
In $D^*(P)$ there remain $r$ non-empty rows and $r$ non-empty columns (for some $r$) and all stable matchings in $D^*(P)$ contain $r$ vertices (that is, they are perfect matchings in $D^*(P)$). 
\end{lemma}

We emphasise that Lemma~\ref{lemma:lem1} holds no matter how unbalanced the original market is and no matter how incomplete preference lists are. As such, any queries concerning the stable matchings of an unbalanced two-sided matching market can be answered by posing the same (reduced) queries of perfect stable matchings in the normal form.
It is not uncommon for researchers to limit attention to balanced markets.
But since the normal form is always balanced, this is not as restrictive an assumption as first appears.

Together Lemma~\ref{lemma:extremalMatches} and Lemma~\ref{lemma:lem1} paint a detailed picture of the structure of the normal form.
In the next section we show how these structural properties can guide us in how to determine the feasibility of assignment constraints.


\section{The algorithm}\label{sec:Algorithm}

The algorithm of this section solves {\sc Reduced Constrained Stable Matchings}. That is, given $D(P)$ and disjoint sets $V^{in}$ and $V^{out}$ of vertices of $D(P)$, return all stable matchings of $D(P)$, each of which includes all vertices in $V^{in}$ and does not contain any vertex in $V^{out}$.
We remind the reader that our reduction in Section~\ref{subsec:reduce} ensures that each worker and each firm appears in at most one pair in $V^{in}$.

For readability, this section is split in two. In Section \ref{sec:Detail} we state precisely the algorithm, and we prove its correctness and time complexity.
In Section \ref{sec:Overview} we give a high level overview of the algorithm by showing how it addresses Reduced Question \ref{q:QuestionR} concerning the matching problem in Reduced Example \ref{exR:initial}.

\subsection{The algorithm:\ in detail}\label{sec:Detail}

We now state the algorithm.
Recall that we assume that the many-to-one environment has been reduced to its equivalent one-to-one formulation.

To find all solutions satisfying the constraints we define a recursive algorithm, {\bf Algorithm~\ref{alg:ALLrec}}, which is called by
our main algorithm, {\bf Algorithm~\ref{alg:ALLmain}}, defined as follows.

\begin{algorithm}\label{alg:ALLmain}
The algorithm has 4 steps.
\begin{description}
 \item[Step \ref{alg:ALLmain}.1.]
 Determine $D^*(P)$ by running IDUA.
 
 \item[Step \ref{alg:ALLmain}.2.]\label{algMain:runIDUA2}
 Let $AllSolutions = \emptyset$ and let $r$ be equal to the number of non-empty rows (or columns) in $D^*(P)$.

 \item[Step \ref{alg:ALLmain}.3.] Call Algorithm~\ref{alg:ALLrec} ($D^*(P)$, $V^{in}$, $V^{out}$)

 \item[Step \ref{alg:ALLmain}.4.] Return $AllSolutions$. 
\end{description}
\end{algorithm}

The following is now the recursive part of the algorithm.

\begin{algorithm}\label{alg:ALLrec} ($D$, $V^{in}$, $V^{out}$)

Note that Algorithm~\ref{alg:ALLrec} is called using the following command: {\bf Algorithm~\ref{alg:ALLrec} ($D$, $V^{in}$, $V^{out}$)}, where
$D$, $V^{in}$ and $V^{out}$ make up the input to the algorithm.
\begin{description}

 \item[Step \ref{alg:ALLrec}.1.]  If {$V^{in}\setminus V(D)\ne \emptyset$}, then return. For all vertices, $(w,f) \in V^{in}$, place all vertices, different from $(w,f)$, 
in the same row or column as $(w,f)$ in $V^{out}$.

 \item[Step \ref{alg:ALLrec}.2.] While there is a vertex  $(w,f) \in V^{out}$ where $d_{\Workers}^+((w,f))=0$ or $d_{\Firms}^+((w,f))=0$,
then modify $D$ by letting it be the digraph obtained by running IDUA on $D - (w,f)$.

 \item[Step \ref{alg:ALLrec}.3.] Now $(w,f) \not\in V^{out}$ for all $(w,f) \in V(D)$ where $d_{\Workers}^+((w,f))=0$ or $d_{\Firms}^+((w,f))=0$. 
Define $M_{\Workers}$ and  $M_{\Firms}$ as follows.

\begin{itemize}
 \item Let $M_{\Workers}$ contain all $(w,f) \in V(D)$ where $d_{\Workers}^+((w,f))=0$.
 \item Let $M_{\Firms}$ contain all $(w,f) \in V(D)$ where $d_{\Firms}^+((w,f))=0$.
\end{itemize}

 \item[Step \ref{alg:ALLrec}.4.] If $|M_{\Workers}|<r$ (which is equivalent to $|M_{\Firms}|<r$ as $|M_{\Workers}|=|M_{\Firms}|$) or if $V^{in} \not \subseteq V(D)$ then return.

 \item[Step \ref{alg:ALLrec}.5.] If $M_{\Workers}=M_{\Firms}$, then add $M_{\Workers}$ to $AllSolutions$ and return.

 \item[Step \ref{alg:ALLrec}.6.] If $M_{\Workers} \not= M_{\Firms}$, then let $(w,f) \in M_{\Workers} \setminus M_{\Firms}$. \big(Such an $(w,f)$ exists since $M_{\Workers} \not= M_{\Firms}$ and $|M_{\Workers}|=|M_{\Firms}|$. Moreover, if there are multiple vertices in $M_{\Workers} \setminus M_{\Firms}$, it does not matter which is chosen.\big) Now make the following two recursive calls before returning.

\begin{description}
\item[Step \ref{alg:ALLrec}.6a.] Call Algorithm~\ref{alg:ALLrec} ($D$, $V^{in} \cup (w,f)$, $V^{out}$).
\item[Step \ref{alg:ALLrec}.6b.] Call Algorithm~\ref{alg:ALLrec} ($D$, $V^{in}$, $V^{out} \cup (w,f)$).

\end{description}
\end{description}
\end{algorithm}

A core idea underlying the algorithm is {that if recursive calls are made in Step~\ref{alg:ALLrec}.6, then the stable matching $M_{\Workers}$ will be found in the recursive call made in Step~\ref{alg:ALLrec}.6a and the 
stable matching $M_{\Firms}$ will be found in the recursive call made in Step~\ref{alg:ALLrec}.6b. Therefore the only time we can possibly return from a call to Algorithm~\ref{alg:ALLrec} without adding a solution to $AllSolutions$ (in Steps \ref{alg:ALLrec}.1 or \ref{alg:ALLrec}.4) is if no solutions exist to our original problem in which case there will be only one call to Algorithm~\ref{alg:ALLrec} (from Algorithm~\ref{alg:ALLmain}).
If $M_{\Workers}=M_{\Firms}$ in Step \ref{alg:ALLrec}.5, then $M_{\Workers}$ is the unique feasible solution in the current call to Algorithm~\ref{alg:ALLrec}.

We will now prove that Algorithm~\ref{alg:ALLmain} returns all stable matchings containing all pairs in $V^{in}$ and not containing any pairs in $V^{out}$ 
in polynomial time per solution. We first prove the correctness of the algorithm before proving its time complexity.
For both results we need Lemma~\ref{lem:NotIn} below that is proved in in Appendix \ref{app:proofs}.
Before stating the lemma, we introduce some notation.
If $D_i$ is a subdigraph of $D^*(P)$ then let $d^+_{\Workers,i}(x)$ {($d^+_{\Firms,i}(x)$, respectively)} denote the number of arcs out of $x \in V(D_i)$ 
in $A_{\Workers}$ {($A_{\Firms}$, respectively)} that are also present in $D_i$.

\begin{lemma} \label{lem:NotIn}
Let $D_i$ be any digraph obtained by running IDUA on some instance of the one-to-one two-sided matching market.
Assume that $D_i$ has $r$ rows and $r$ columns (and therefore, by Lemma~\ref{lemma:lem1}, all stable matchings are perfect matchings of size $r$).
Let $v \in V(D_i)$ be arbitrary, such that $d_{\Workers,i}^+(v)=0$ {or $d_{\Firms,i}^+(v)=0$} in $D_i$.
Then, the following two conditions are equivalent.
\begin{enumerate}[label=(\roman*)]
\item
$M$ is a stable matching in $D_i$ not containing $v$.

\item
$M$ is a stable matching in $D_i - v$.
\end{enumerate}
\end{lemma}

Let us now provide intuition for Lemma~\ref{lem:NotIn} and show how it sheds further light on the structure of a matching market.
We then explain its importance for the algorithm.

Suppose first that $D_{i}$ in the statement of the lemma represents the normal form, and recall that the normal form is a balanced market where the size of each stable matching is given by $r$.
Now, consider some pair $v = (\worker{}, \firm{}) \in D_{i}$ that (i) is forbidden to be part of any stable matching, and also (ii) has out degree zero in one direction but not the other. 
By Lemma~\ref{lemma:extremalMatches}, $v$ is contained in the optimal stable matching for one party, say worker $\worker{}$.
Moreover, due to the tension in preferences over stable matchings that exists between the two sides, $v$ must also be contained in the least desirable stable matchings for $\firm{}$ since $\worker{}$ is $\firm{}$'s least preferred worker in the normal form.
The key insight of Lemma~\ref{lem:NotIn} is that $v$ cannot be a blocking pair.
(In fact, no extreme pair can ever be blocking.)
To see why, consider a matching in the normal form in which $\firm{}$ is paired with some other worker, say $\worker{}' \neq \worker{}$.
It must be that $\firm{}$ prefers $\worker{}'$ over $\worker{}$, and so $\firm{}$ would never leave the pair $(\worker{}', \firm{})$ for the pair $(\worker{}, \firm{})$.

Given that extreme pair $v = (\worker{}, \firm{})$ cannot be blocking against any matching in the normal form, it can be removed without impacting the stable matchings that $v$ is not a part of.
That is, while deleting $v$ will certainly remove some matchings, it will only remove those stable matchings that contain $v$ and these are stable matchings that we want to discard.
Since $v$ was the optimal feasible pair for $\worker{}$, with $v$ removed it must be that $\worker{}$ has a new favourite firm.
Given this, it is possible that the new market can be reduced further by running IDUA again (and recall that running IDUA never changes the set of stable matchings).
Lastly, note that $D_i$ need not refer only to the normal form, so this procedure can be used ad infinitum.

Armed with Lemma~\ref{lem:NotIn} we can now state the main theorem, the proof of which is found in Appendix \ref{app:proofs}.
Before stating the theorem, we introduce some notation.
Recall that $m$ is the number of workers and $n$ the number of positions in $P$.
Let $N=\max\{m,n\}$ and let $M=\min\{m,n\}$.
We denote by $s$ the number of stable matchings outputted by the algorithm.

\begin{theorem}\label{thm:ALLmain}
Algorithm \ref{alg:ALLmain} correctly solves {\sc Reduced Constrained Stable Matchings} and does so in time $O(N^3 + N^3s)$.  Moreover, the time between two consecutive stable matchings outputted by the algorithm is $O(N^3)$.
\end{theorem}

We conclude this section by recalling that if a stable matching satisfying the constraints stipulated by {\sc Reduced Constrained Stable Matchings} exists, then Algorithm \ref{alg:ALLmain} will not only confirm but will output {\it all} of them.
Of course it is possible that a market designer might be content with a single stable matching satisfying the constraints.
This is not a problem.
It is straightforward to modify our algorithm in order to output either side's optimal stable matching satisfying the constraints.
{One simply has to output $M_{\Workers}$ or $M_{\Firms}$ in the first call to Algorithm~\ref{alg:ALLrec} instead of making recursive calls in Step~\ref{alg:ALLrec}.6.}

\subsection{The algorithm:\ an example}\label{sec:Overview}

Let us now illustrate how our algorithm works.
To do so we use the two-sided market from Reduced Example~\ref{exR:initial} and the market designer's Reduced Question} \ref{q:QuestionR} pertaining to it.

The input to Algorithm~\ref{alg:ALLmain} will be the preference lists from Reduced Example~\ref{exR:initial}.
After the reduction to the one-to-one environment has been performed this is equivalent to the digraph $D(P)$ in Figure~\ref{fig:ex1}. The variant of {\sc Reduced Constrained Stable Matchings} submitted to Algorithm~\ref{alg:ALLrec} is {Reduced Question} \ref{q:QuestionR}:\ decide if there exists a stable matching in $D(P)$ containing the vertex $(\worker{1},\firm{2})$ and also avoiding the vertex $(\worker{4},\firm{1})$.
The output will be the empty set if there is no such stable matching and all such stable matchings in the event that one exists.

Step~\ref{alg:ALLmain}.1 of  Algorithm~\ref{alg:ALLmain} runs IDUA to generate the normal form.
The result of this is depicted in Figure~\ref{fig:ex2} where the colour coding of vertices is as follows:\ vertices are coloured {green} if they are required to be in the stable matching and coloured {red} if required not to be. 
Once the IDUA procedure terminates and the normal form $D^*(P)$ has been arrived at, Algorithm~\ref{alg:ALLmain} then calls Algorithm~\ref{alg:ALLrec}.

The first step of Algorithm~\ref{alg:ALLrec}, Step~\ref{alg:ALLrec}.1, makes two stipulations.
In the first, it checks if any of the  \WithColor{green}\WithoutColor{white} vertices were deleted when applying IDUA. If so, then there is no stable matching satisfying the requirements and the algorithm exits.
Note that upon running IDUA for the example, vertex $(1,2)$ in $D(P)$ in the left hand panel of Figure~\ref{fig:ex2} remains in $D^*(P)$ in the right hand panel.
We observe that two of the red vertices, vertices $(6,4)$ and $(6,5)$, that were required not to be part of a desired stable matching are no longer present in $D^*(P)$ in the right hand panel of Figure~\ref{fig:ex2}.
This simply confirms that these vertices are not part of any stable matching, and so a constraint requiring that neither be contained in a particular stable matching is rendered redundant. To say it another way, if red vertices are deleted in going from $D(P)$ to $D^*(P)$, the algorithm proceeds unaffected.

Step~\ref{alg:ALLrec}.1 also converts all queries about required pairs (green vertices) into corresponding queries about avoided pairs (red vertices). Since the green vertex, $(1, 2)$, remains after applying IDUA, Lemma~\ref{lemma:lem1} confirms that worker $\worker{1}$ is employed in every stable matching and that the single position at $\firm{2}$ is filled in every stable matching. But given this, the query about worker $\worker{1}$ being employed at firm $\firm{2}$ can be converted into a collection of equivalent queries about worker $\worker{1}$ not being employed at any firm other than $\firm{2}$ and firm $\firm{2}$'s position not being filled by any worker other than $\worker{1}$.
Clearly any stable matching that avoids all these pairs must contain the pair $(\worker{1}, \firm{2})$.

Thus the problem of determining if there exists a stable matching in which worker $\worker{1}$ is employed at firm $\firm{2}$ is equivalent to one of determining if there is a stable matching in the normal form that does not include vertices $(1,1), (1,3), (1,5)$ and vertices $(2,2), (3,2), (4,2)$, where the first collection of three vertices require that $\worker{1}$ must not be employed at $\firm{1}, \firm{3}$, or $\firm{5}$ and the second collection of three vertices require that the position at $\firm{2}$ is not filled by $\worker{2}, \worker{3}$, or $\worker{4}$.
Panel (a) in Figure~\ref{fig:ex3} displays the result of converting the query about the required (green) vertex $(1,2)$ in $D^*(P)$, shown in panel (b) of Figure~\ref{fig:ex2}, to an analogous query about (red) vertices being avoided in $D^*(P)$.

\begin{figure}[h!]
\begin{center}
\fbox{
\begin{minipage}{5.8cm}
\tikzstyle{vertexDOT}=[scale=0.12,circle,draw,fill]
\tikzstyle{vertexY}=[circle,draw, top color=gray!10, bottom color=gray!40, minimum size=11pt, scale=0.38, inner sep=0.99pt]

\WithColor{\tikzstyle{vertexIn}=[circle,draw, top color=green!20, bottom color=green!60, minimum size=11pt, scale=0.38, inner sep=0.99pt]}
\WithoutColor{\tikzstyle{vertexIn}=[circle,draw, top color=gray!1, bottom color=gray!1, minimum size=21pt, scale=0.38, inner sep=2.99pt]}

\WithColor{\tikzstyle{vertexOut}=[circle,draw, top color=red!20, bottom color=red!60, minimum size=11pt, scale=0.38, inner sep=0.99pt]}
\tikzstyle{vertexPivot}=[circle,draw, top color=black!40, bottom color=black!80, minimum size=11pt, scale=0.38, inner sep=0.99pt]
\tikzstyle{vertexDelete}=[circle,dotted, draw, top color=gray!1, bottom color=gray!1, minimum size=11pt, scale=0.3, inner sep=0.99pt]
\WithoutColor{\tikzstyle{vertexOut}=[circle,draw, top color=gray!99, bottom color=gray!99, minimum size=21pt, scale=0.38, inner sep=2.99pt]}
\begin{tikzpicture}[scale=0.67]
\node (x11) at (1,11) [vertexOut] {$(1,1)$};
\node (x12) at (3,11) [vertexIn] {$(1,2)$};
\node (x13) at (5,11) [vertexOut] {$(1,3)$};
\node (x15) at (9,11) [vertexOut] {$(1,5)$};

\node (x21) at (1,9) [vertexY] {$(2,1)$};
\node (x22) at (3,9) [vertexOut] {$(2,2)$};
\node (x23) at (5,9) [vertexY] {$(2,3)$};
\node (x25) at (9,9) [vertexY] {$(2,5)$};

\node (x31) at (1,7) [vertexY] {$(3,1)$};
\node (x32) at (3,7) [vertexOut] {$(3,2)$};
\node (x33) at (5,7) [vertexY] {$(3,3)$};
\node (x35) at (9,7) [vertexY] {$(3,5)$};

\node (x41) at (1,5) [vertexOut] {$(4,1)$};
\node (x42) at (3,5) [vertexOut] {$(4,2)$};
\node (x43) at (5,5) [vertexY] {$(4,3)$};
\node (x45) at (9,5) [vertexY] {$(4,5)$};

\node (x54) at (7,3) [vertexY] {$(5,4)$};

\draw [->,line width=0.03cm] (x15) -- (x13);
\draw [->,line width=0.03cm] (x13) -- (x12);
\draw [->,line width=0.03cm] (x12) -- (x11);

\draw [->,line width=0.03cm] (x23) -- (x25);
\draw [->,line width=0.03cm] (x25) to [out=200,in=340] (x21);
\draw [->,line width=0.03cm] (x21) -- (x22);

\draw [->,line width=0.03cm] (x32) -- (x31);
\draw [->,line width=0.03cm] (x31) to [out=20,in=160] (x35);
\draw [->,line width=0.03cm] (x35) -- (x33);

\draw [->,line width=0.03cm] (x41) -- (x42);
\draw [->,line width=0.03cm] (x42) -- (x43);
\draw [->,line width=0.03cm] (x43) -- (x45);

\draw [->,line width=0.03cm] (x11) -- (x21);
\draw [->,line width=0.03cm] (x21) -- (x31);
\draw [->,line width=0.03cm] (x31) -- (x41);

\draw [->,line width=0.03cm] (x22) -- (x12);
\draw [->,line width=0.03cm] (x12) to [out=250,in=110] (x42);
\draw [->,line width=0.03cm] (x42) -- (x32);

\draw [->,line width=0.03cm] (x33) -- (x43);
\draw [->,line width=0.03cm] (x43) to [out=70,in=290] (x13);
\draw [->,line width=0.03cm] (x13) -- (x23);

\draw [->,line width=0.03cm] (x45) -- (x35);
\draw [->,line width=0.03cm] (x35) -- (x25);
\draw [->,line width=0.03cm] (x25) -- (x15);
\node [scale=0.7] at (5.0,1.5) {(a) $D^{*}(P)$ after converting queries about};
\node [scale=0.7] at (5.0,1.0) { $V^{in}$ (green) into queries about $V^{out}$ (red).};

\end{tikzpicture}
\end{minipage} } \hspace{0.2cm} \fbox{
\begin{minipage}{5.8cm}
\tikzstyle{vertexDOT}=[scale=0.12,circle,draw,fill]
\tikzstyle{vertexY}=[circle,draw, top color=gray!10, bottom color=gray!35, minimum size=11pt, scale=0.38, inner sep=0.99pt]

\WithColor{\tikzstyle{vertexIn}=[circle,draw, top color=green!20, bottom color=green!60, minimum size=11pt, scale=0.38, inner sep=0.99pt]}
\WithoutColor{\tikzstyle{vertexIn}=[circle,draw, top color=gray!1, bottom color=gray!1, minimum size=21pt, scale=0.38, inner sep=2.99pt]}

\WithColor{\tikzstyle{vertexOut}=[circle,draw, top color=red!20, bottom color=red!60, minimum size=11pt, scale=0.38, inner sep=0.99pt]}
\tikzstyle{vertexPivot}=[circle,draw, top color=black!40, bottom color=black!80, minimum size=11pt, scale=0.38, inner sep=0.99pt]
\tikzstyle{vertexDelete}=[circle,dotted, draw, top color=gray!1, bottom color=gray!1, minimum size=11pt, scale=0.3, inner sep=0.99pt]
\WithoutColor{\tikzstyle{vertexOut}=[circle,draw, top color=gray!99, bottom color=gray!99, minimum size=21pt, scale=0.38, inner sep=2.99pt]}

\begin{tikzpicture}[scale=0.67]
\node (x12) at (3,11) [vertexIn] {$(1,2)$};

\node (x21) at (1,9) [vertexY] {$(2,1)$};
\node (x25) at (9,9) [vertexY] {$(2,5)$};

\node (x31) at (1,7) [vertexY] {$(3,1)$};
\node (x33) at (5,7) [vertexY] {$(3,3)$};
\node (x35) at (9,7) [vertexY] {$(3,5)$};

\node (x43) at (5,5) [vertexY] {$(4,3)$};
\node (x45) at (9,5) [vertexY] {$(4,5)$};

\node (x54) at (7,3) [vertexY] {$(5,4)$};


\draw [->,line width=0.03cm] (x25) -- (x21);

\draw [->,line width=0.03cm] (x31) to [out=20,in=160] (x35);
\draw [->,line width=0.03cm] (x35) -- (x33);

\draw [->,line width=0.03cm] (x43) -- (x45);

\draw [->,line width=0.03cm] (x21) -- (x31);


\draw [->,line width=0.03cm] (x33) -- (x43);

\draw [->,line width=0.03cm] (x45) -- (x35);
\draw [->,line width=0.03cm] (x35) -- (x25);

\node [scale=0.7] at (5.0,1.5) {(b) Result of applying Step 2.2.};
\node [scale=0.7] at (5.0,1.0) {Note that no red vertices remain.};

\end{tikzpicture}
\end{minipage} }

\end{center}

\caption{In (a) we depict $D^*(P)$ having converted the problem into one of avoiding vertices. In (b) we depict the submatching problem that is arrived at after completing Step~\ref{alg:ALLrec}.2 in the first call to Algorithm~\ref{alg:ALLrec}.}
\label{fig:ex3}
\end{figure}
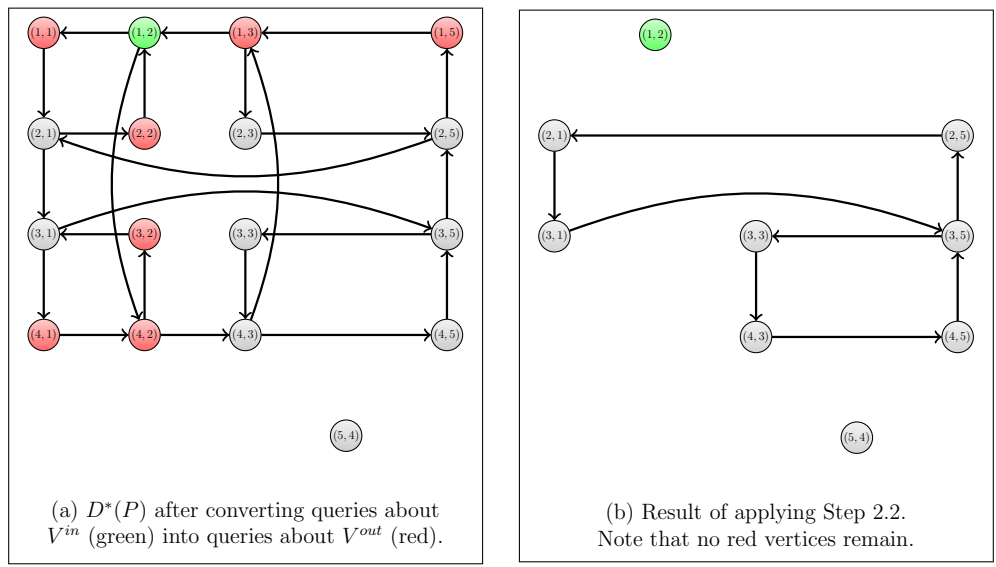

Step~\ref{alg:ALLrec}.2 is one of the key steps.
It states that whenever there is a red vertex with either horizontal out-degree zero or vertical out-degree zero, that this vertex is deleted from the digraph and the IDUA procedure is run on the subdigraph that remains.
Recall that deleting red vertices with either horizontal out-degree zero or vertical out-degree zero is justified by Lemma~\ref{lem:NotIn}.

In our example, the vertices satisfying this requirement can be seen in panel (a) of Figure~\ref{fig:ex3}.
They are $(1,1), (2,2), (1,5), (3,2)$, and $(4,1)$, where the first two vertices have horizontal out-degree equal to zero and the latter three vertices have vertical out-degree equal to zero.
While any of the vertices satisfying this constraint can be selected at point Step~\ref{alg:ALLrec}.2 -- the order is immaterial -- we note that not all red vertices need satisfy the constraint.
For our example, from panel (a) of Figure~\ref{fig:ex3} it is clear that neither vertex $(1,3)$ nor vertex $(4,2)$ satisfy the requirement as both have positive horizontal out-degree and positive vertical out-degree.
We emphasise this point strongly as it is not possble simply to run IDUA on $D^* - V^{out}$.\footnote{In Appendix \ref{app:exampleDelete} we present an example that shows that deleting all vertices in $V^{out}$ and then applying IDUA may yield an incorrect conclusion.}

Panel (b) of Figure~\ref{fig:ex3} depicts the submatching problem that is arrived at after completing Step~\ref{alg:ALLrec}.2 in the first call to Algorithm~\ref{alg:ALLrec}.
By definition there will be no red vertices with out-degree zero remaining.
(Though we recall that red vertices with positive out-degree could in theory remain, albeit none did in this example. That is, the vertices $(1, 3)$ and $(4,2)$ were removed when the IDUA procedure was applied to some reduced subdigraph; they were not directly called by Step~\ref{alg:ALLrec}.2.)
Moreover, this step can also delete other vertices. For example, in going from panel (a) to panel (b) in Figure~\ref{fig:ex3}, we note that red vertices $(1,3)$ and $(4,2)$ and the non-red vertex $(2,3)$ were all deleted by Step~\ref{alg:ALLrec}.2.
In the proof of correctness of Algorithm~\ref{alg:ALLrec} in Theorem \ref{thm:ALLmain} we show that the order in which red vertices satisfying the requirement of Step~\ref{alg:ALLrec}.2 are deleted and then apply IDUA does not impact the final outcome.

Step~\ref{alg:ALLrec}.3 now collects the worker-optimal pairs for every worker, $M_{\Workers}$, and the firm-optimal pairs for every firm, $M_{\Firms}$. Step~\ref{alg:ALLrec}.4 checks to see that both $M_{\Workers}$ and $M_{\Firms}$ are the size of a stable matching.
If not, then there cannot be a stable matching satisfying the constraints and the Algorithm states as much.
If both $M_{\Workers}$ and $M_{\Firms}$ are the size of a stable matching, then Step~\ref{alg:ALLrec}.5 checks to see if they are one and the same.
If they are, that is if $M_{\Workers} = M_{\Firms}$, then there is exactly one stable matching satisfying the constraints and the algorithm returns it.

If $M_{\Workers} \neq M_{\Firms}$, then we can make two observations.
First, there are multiple (at least two) stable matchings satisfying the constraints.
Moreover, given our result on uniqueness in \cite{GutinNeary:2023:GEB}, there must be some cycles in the remaining digraph because cycles that cannot be broken by the IDUA procedure are the barrier to uniqueness.
We observe that there are two cycles of length 4 in panel (b) of Figure~\ref{fig:ex3} and Figure~\ref{fig:extremeConstrained}.
Second, it must be that $M_{\Workers}$ ($M_{\Firms}$) is the worker-optimal (firm-optimal) stable matching satisfying the constraints.
For the market of Reduced Example~\ref{exR:initial} we depict $M_{\Workers}$ and $M_{\Firms}$ in Figure~\ref{fig:extremeConstrained} below.
We adopt the same colour coding as that in Figure~\ref{fig:exExtremeMatchings}.
That is, vertices solely in the worker-optimal stable matching satisfying the constraints are coloured yellow, vertices solely in the firm-optimal stable matching satisfying the constraints are coloured blue, and vertices in both are coloured a hybrid of yellow and blue.

\begin{figure}[h!]
\begin{center}

\tikzstyle{vertexDOT}=[scale=0.23,circle,draw,fill]
\tikzstyle{vertexY}=[circle,draw, top color=gray!10, bottom color=gray!40, minimum size=11pt, scale=0.5, inner sep=0.99pt]

\WithColor{\tikzstyle{vertexZ}=[circle,draw, top color=\matchingC{}!20, bottom color=\matchingC{}!50, minimum size=11pt, scale=0.5, inner sep=0.99pt]} 
\tikzstyle{vertexWF}=[circle,draw, top color=blue!70, bottom color=yellow!100, minimum size=11pt, scale=0.5, inner sep=0.99pt]
\tikzstyle{vertexF}=[circle,draw, top color=blue!50, bottom color=blue!50, minimum size=11pt, scale=0.5, inner sep=0.99pt]
\WithoutColor{\tikzstyle{vertexZ}=[rectangle,draw, top color=gray!1, bottom color=gray!1, minimum size=21pt, scale=0.6, inner sep=2.99pt]}  

\tikzstyle{vertexW}=[circle,draw, top color=gray!1, bottom color=gray!1, minimum size=11pt, scale=0.4, inner sep=0.99pt]
\tikzstyle{vertexQ}=[circle,draw, top color=black!99, bottom color=black!99, minimum size=11pt, scale=0.4, inner sep=0.99pt]


%
%

\begin{tikzpicture}[scale=0.8]

\node [scale=0.9] at (-1,11) {$w_1$};
\node [scale=0.9] at (-1,9) {$w_2$};
\node [scale=0.9] at (-1,7) {$w_3$};
\node [scale=0.9] at (-1,5) {$w_4$};
\node [scale=0.9] at (-1,3) {$w_5$};

\node [scale=0.9] at (1,12.5) {$f_1$};
\node [scale=0.9] at (3,12.5) {$f_2$};
\node [scale=0.9] at (5,12.5) {$f_3$};
\node [scale=0.9] at (7,12.5) {$f_4^1$};
\node [scale=0.9] at (9,12.5) {$f_4^2$};
\node (x12) at (3,11) [vertexWF] {$(1,2)$};

\node (x21) at (1,9) [vertexZ] {$(2,1)$};
\node (x25) at (9,9) [vertexF] {$(2,5)$};

\node (x31) at (1,7) [vertexF] {$(3,1)$};
\node (x33) at (5,7) [vertexZ] {$(3,3)$};
\node (x35) at (9,7) [vertexY] {$(3,5)$};

\node (x43) at (5,5) [vertexF] {$(4,3)$};
\node (x45) at (9,5) [vertexZ] {$(4,5)$};

\node (x54) at (7,3) [vertexWF] {$(5,4)$};


\draw [->,line width=0.03cm] (x25) -- (x21);

\draw [->,line width=0.03cm] (x31) to [out=20,in=160] (x35);
\draw [->,line width=0.03cm] (x35) -- (x33);

\draw [->,line width=0.03cm] (x43) -- (x45);

\draw [->,line width=0.03cm] (x21) -- (x31);


\draw [->,line width=0.03cm] (x33) -- (x43);

\draw [->,line width=0.03cm] (x45) -- (x35);
\draw [->,line width=0.03cm] (x35) -- (x25);


\end{tikzpicture}
\end{center}
\caption{The two extremal stable matchings satisfying the constraints.}
\label{fig:extremeConstrained}
\end{figure}
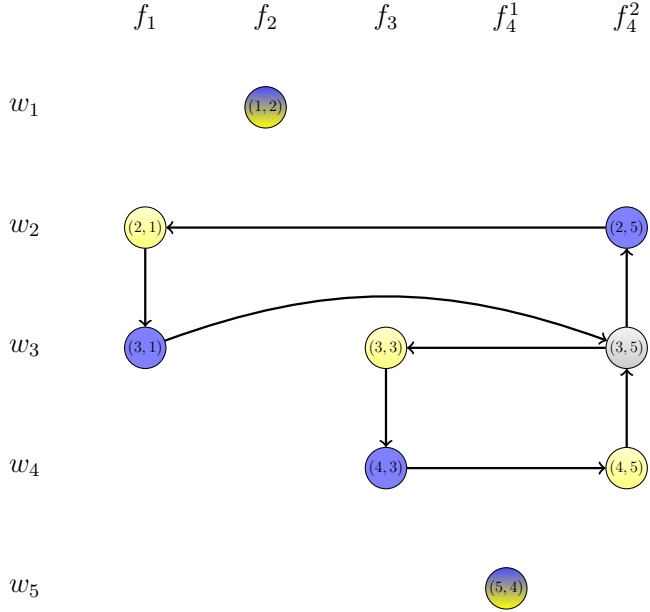

Since $M_{\Workers{}} \neq M_{\Firms{}}$, as seen in Figure~\ref{fig:extremeConstrained}, then Step~\ref{alg:ALLrec}.6 of the algorithm is called. 
The first thing this step does is to identify {a vertex, $(\worker{}, \firm{})$,  that is contained in $M_{\Workers{}}$ but not in $M_{\Firms{}}$.} 
(Given that $M_{\Workers{}} \neq M_{\Firms{}}$ there must be at least one such vertex, and if there is more than one such vertex which is selected is immaterial.) 
We now run the recursive part of Algorithm~\ref{alg:ALLrec} on two subcases:\ one where $(\worker{}, \firm{})$ is required to be included (so that $(\worker{}, \firm{})$ is added to $V^{in}$) and the other where it is required not to be included (and hence $(\worker{}, \firm{})$ is contained in $V^{out}$).
As can be seen from Figure~\ref{fig:extremeConstrained}, running left to right and down the page the first vertex we encounter of this form is the vertex $(2, 1)$.
Thus we now run Algorithm~\ref{alg:ALLrec} on $(D, V^{in} \cup (2,1), V^{out})$ and on $(D, V^{in}, V^{out} \cup (2,1))$.

This repeated splitting into two subcases generates a binary search-tree that we depict in Figure~\ref{fig:SearchTree}. There are five pairs of panels in Figure~\ref{fig:SearchTree} labelled (a) through (e).
In every pair of panels, the left panel depicts the temporary constraints while the right hand panel depicts the extremal matchings associated with the temporary constraints.
A path through the tree ends in a pair of panels that provide a stable matching satisfying the constraints. 
Paths through the tree can be of different length (as indeed they are in Figure~\ref{fig:SearchTree}).

\begin{figure}[h!]

\tikzstyle{vertexDOT}=[scale=0.12,circle,draw,fill]
\tikzstyle{vertexY}=[circle,draw, top color=gray!20, bottom color=gray!45, minimum size=11pt, scale=0.2, inner sep=0.99pt]

\WithColor{\tikzstyle{vertexIn}=[circle,draw, top color=green!20, bottom color=green!60, minimum size=11pt, scale=0.2, inner sep=0.99pt]}
\WithoutColor{\tikzstyle{vertexIn}=[circle,draw, top color=gray!1, bottom color=gray!1, minimum size=21pt, scale=0.2, inner sep=2.99pt]}

\WithColor{\tikzstyle{vertexOut}=[circle,draw, top color=red!20, bottom color=red!60, minimum size=11pt, scale=0.2, inner sep=0.99pt]}
\WithoutColor{\tikzstyle{vertexOut}=[circle,draw, top color=gray!99, bottom color=gray!99, minimum size=21pt, scale=0.2, inner sep=2.99pt]}

\WithColor{\tikzstyle{vertexWF}=[circle,draw, top color=blue!70, bottom color=yellow!100, minimum size=11pt, scale=0.2, inner sep=0.99pt]}
\WithoutColor{\tikzstyle{vertexWF}=[circle,draw, top color=black!70, bottom color=white!100, minimum size=21pt, scale=0.2, inner sep=2.99pt]}

\WithColor{\tikzstyle{vertexWW}=[circle,draw, top color=yellow!100, bottom color=yellow!50, minimum size=11pt, scale=0.2, inner sep=0.99pt]}
\WithoutColor{\tikzstyle{vertexWW}=[circle,draw, top color=black!100, bottom color=black!100, minimum size=21pt, scale=0.2, inner sep=2.99pt]}

\WithColor{\tikzstyle{vertexFF}=[circle,draw, top color=blue!100, bottom color=blue!50, minimum size=11pt, scale=0.2, inner sep=0.99pt]}
\WithoutColor{\tikzstyle{vertexFF}=[circle,draw, top color=white!100, bottom color=white!100, minimum size=21pt, scale=0.2, inner sep=2.99pt]}

\begin{tikzpicture}[scale=0.38]

\node [scale=1.0] at (13,51.5) {(a)};

\node (x12b) at (12,50) [vertexIn] {$(1,2)$};

\node (x21b) at (10,48) [vertexY] {$(2,1)$};
\node (x25b) at (16,48) [vertexY] {$(2,5)$};

\node (x31b) at (10,46) [vertexY] {$(3,1)$};
\node (x33b) at (14,46) [vertexY] {$(3,3)$};
\node (x35b) at (16,46) [vertexY] {$(3,5)$};

\node (x43b) at (14,44) [vertexY] {$(4,3)$};
\node (x45b) at (16,44) [vertexY] {$(4,5)$};

\node (x54b) at (15,43) [vertexY] {$(5,4)$};

\draw [->,line width=0.03cm] (x25b) -- (x21b);
\draw [->,line width=0.03cm] (x31b) to [out=20,in=160] (x35b);
\draw [->,line width=0.03cm] (x35b) -- (x33b);
\draw [->,line width=0.03cm] (x43b) -- (x45b);
\draw [->,line width=0.03cm] (x21b) -- (x31b);
\draw [->,line width=0.03cm] (x33b) -- (x43b);
\draw [->,line width=0.03cm] (x45b) -- (x35b);
\draw [->,line width=0.03cm] (x35b) -- (x25b);


\WithColor{ \node [scale=1.0] at (22,51.5) {{\color{yellow} $M_{\Workers}$} and {\color{blue} $M_{\Firms}$}}; }
\WithoutColor{ \node [scale=1.1] at (22,51.5) {$M_{\Workers}$ and $M_{\Firms}$}; }

\node (x12c) at (21,50) [vertexWF] {$(1,2)$};

\node (x21c) at (19,48) [vertexWW] {$(2,1)$};
\node (x25c) at (25,48) [vertexFF] {$(2,5)$};

\node (x31c) at (19,46) [vertexFF] {$(3,1)$};
\node (x33c) at (23,46) [vertexWW] {$(3,3)$};
\node (x35c) at (25,46) [vertexY] {$(3,5)$};

\node (x43c) at (23,44) [vertexFF] {$(4,3)$};
\node (x45c) at (25,44) [vertexWW] {$(4,5)$};

\node (x54c) at (24,43) [vertexWF] {$(5,4)$};

\draw [->,line width=0.03cm] (x25c) -- (x21c);
\draw [->,line width=0.03cm] (x31c) to [out=20,in=160] (x35c);
\draw [->,line width=0.03cm] (x35c) -- (x33c);
\draw [->,line width=0.03cm] (x43c) -- (x45c);
\draw [->,line width=0.03cm] (x21c) -- (x31c);
\draw [->,line width=0.03cm] (x33c) -- (x43c);
\draw [->,line width=0.03cm] (x45c) -- (x35c);
\draw [->,line width=0.03cm] (x35c) -- (x25c);

\draw [line width=0.03cm] (26,53) -- (26,41) -- (8.5,41) -- (8.5,53) -- (26,53);
\draw [line width=0.03cm] (8.5,53) -- (8.5,41);
\draw [line width=0.03cm] (17.5,53) -- (17.5,41);


\node [scale=1.0] at (4,34.5) {(b)};

\node (x12d) at (3,33) [vertexIn] {$(1,2)$};

\node (x21d) at (1,31) [vertexIn] {$(2,1)$};
\node (x25d) at (7,31) [vertexOut] {$(2,5)$};

\node (x31d) at (1,29) [vertexOut] {$(3,1)$};
\node (x33d) at (5,29) [vertexY] {$(3,3)$};
\node (x35d) at (7,29) [vertexY] {$(3,5)$};

\node (x43d) at (5,27) [vertexY] {$(4,3)$};
\node (x45d) at (7,27) [vertexY] {$(4,5)$};

\node (x54d) at (6,26) [vertexY] {$(5,4)$};

\draw [->,line width=0.03cm] (x25d) -- (x21d);
\draw [->,line width=0.03cm] (x31d) to [out=20,in=160] (x35d);
\draw [->,line width=0.03cm] (x35d) -- (x33d);
\draw [->,line width=0.03cm] (x43d) -- (x45d);
\draw [->,line width=0.03cm] (x21d) -- (x31d);
\draw [->,line width=0.03cm] (x33d) -- (x43d);
\draw [->,line width=0.03cm] (x45d) -- (x35d);
\draw [->,line width=0.03cm] (x35d) -- (x25d);


\WithColor{ \node [scale=1.0] at (12,34.5) {{\color{yellow} $M_{\Workers}$} and {\color{blue} $M_{\Firms}$}}; }
\WithoutColor{ \node [scale=1.1] at (12,34.5) {$M_{\Workers}$ and $M_{\Firms}$}; }

\node (x12e) at (11,33) [vertexWF] {$(1,2)$};

\node (x21e) at (9,31) [vertexWF] {$(2,1)$};

\node (x33e) at (13,29) [vertexWW] {$(3,3)$};
\node (x35e) at (15,29) [vertexFF] {$(3,5)$};

\node (x43e) at (13,27) [vertexFF] {$(4,3)$};
\node (x45e) at (15,27) [vertexWW] {$(4,5)$};

\node (x54e) at (14,26) [vertexWF] {$(5,4)$};

\draw [->,line width=0.03cm] (x35e) -- (x33e);
\draw [->,line width=0.03cm] (x43e) -- (x45e);
\draw [->,line width=0.03cm] (x33e) -- (x43e);
\draw [->,line width=0.03cm] (x45e) -- (x35e);

\draw [line width=0.03cm] (16,36) -- (16,24) -- (0,24) -- (0,36) -- (16,36);
\draw [line width=0.03cm] (8,36) -- (8,24);


\node [scale=1.0] at (24,34.5) {(c)};

\node (x12f) at (23,33) [vertexIn] {$(1,2)$};

\node (x21f) at (21,31) [vertexOut] {$(2,1)$};
\node (x25f) at (27,31) [vertexY] {$(2,5)$};

\node (x31f) at (21,29) [vertexY] {$(3,1)$};
\node (x33f) at (25,29) [vertexY] {$(3,3)$};
\node (x35f) at (27,29) [vertexY] {$(3,5)$};

\node (x43f) at (25,27) [vertexY] {$(4,3)$};
\node (x45f) at (27,27) [vertexY] {$(4,5)$};

\node (x54f) at (26,26) [vertexY] {$(5,4)$};

\draw [->,line width=0.03cm] (x25f) -- (x21f);
\draw [->,line width=0.03cm] (x31f) to [out=20,in=160] (x35f);
\draw [->,line width=0.03cm] (x35f) -- (x33f);
\draw [->,line width=0.03cm] (x43f) -- (x45f);
\draw [->,line width=0.03cm] (x21f) -- (x31f);
\draw [->,line width=0.03cm] (x33f) -- (x43f);
\draw [->,line width=0.03cm] (x45f) -- (x35f);
\draw [->,line width=0.03cm] (x35f) -- (x25f);


\WithColor{ \node [scale=1.0] at (32,34.5) {{\color{yellow} $M_{\Workers}$} and {\color{blue} $M_{\Firms}$}}; }
\WithoutColor{ \node [scale=1.1] at (32,34.5) {$M_{\Workers}$ and $M_{\Firms}$}; }

\node (x12g) at (31,33) [vertexWF] {$(1,2)$};

\node (x25g) at (35,31) [vertexWF] {$(2,5)$};

\node (x31g) at (29,29) [vertexWF] {$(3,1)$};

\node (x43g) at (33,27) [vertexWF] {$(4,3)$};

\node (x54g) at (34,26) [vertexWF] {$(5,4)$};

\draw [line width=0.03cm] (36,36) -- (36,24) -- (20,24) -- (20,36) -- (36,36);
\draw [line width=0.03cm] (28,36) -- (28,24);


\node [scale=1.0] at (4,17.5) {(d)};

\node (x12h) at (3,16) [vertexIn] {$(1,2)$};

\node (x21h) at (1,14) [vertexIn] {$(2,1)$};

\node (x33h) at (5,12) [vertexIn] {$(3,3)$};
\node (x35h) at (7,12) [vertexOut] {$(3,5)$};

\node (x43h) at (5,10) [vertexOut] {$(4,3)$};
\node (x45h) at (7,10) [vertexY] {$(4,5)$};

\node (x54h) at (6,9) [vertexY] {$(5,4)$};

\draw [->,line width=0.03cm] (x35h) -- (x33h);
\draw [->,line width=0.03cm] (x43h) -- (x45h);
\draw [->,line width=0.03cm] (x33h) -- (x43h);
\draw [->,line width=0.03cm] (x45h) -- (x35h);


\WithColor{ \node [scale=1.0] at (12,17.5) {{\color{yellow} $M_{\Workers}$} and {\color{blue} $M_{\Firms}$}}; }
\WithoutColor{ \node [scale=1.1] at (12,17.5) {$M_{\Workers}$ and $M_{\Firms}$}; }

\node (x12i) at (11,16) [vertexWF] {$(1,2)$};

\node (x21i) at (9,14) [vertexWF] {$(2,1)$};

\node (x33i) at (13,12) [vertexWF] {$(3,3)$};

\node (x45i) at (15,10) [vertexWF] {$(4,5)$};

\node (x54i) at (14,9) [vertexWF] {$(5,4)$};

\draw [->,line width=0.03cm] (x35e) -- (x33e);
\draw [->,line width=0.03cm] (x43e) -- (x45e);
\draw [->,line width=0.03cm] (x33e) -- (x43e);
\draw [->,line width=0.03cm] (x45e) -- (x35e);

\draw [line width=0.03cm] (16,19) -- (16,7) -- (0,7) -- (0,19) -- (16,19);
\draw [line width=0.03cm] (8,19) -- (8,7);


\node [scale=1.0] at (24,17.5) {(e)};

\node (x12j) at (23,16) [vertexIn] {$(1,2)$};

\node (x21j) at (21,14) [vertexIn] {$(2,1)$};

\node (x33j) at (25,12) [vertexOut] {$(3,3)$};
\node (x35j) at (27,12) [vertexY] {$(3,5)$};

\node (x43j) at (25,10) [vertexY] {$(4,3)$};
\node (x45j) at (27,10) [vertexY] {$(4,5)$};

\node (x54j) at (26,9) [vertexY] {$(5,4)$};

\draw [->,line width=0.03cm] (x35j) -- (x33j);
\draw [->,line width=0.03cm] (x43j) -- (x45j);
\draw [->,line width=0.03cm] (x33j) -- (x43j);
\draw [->,line width=0.03cm] (x45j) -- (x35j);


\WithColor{ \node [scale=1.0] at (32,17.5) {{\color{yellow} $M_{\Workers}$} and {\color{blue} $M_{\Firms}$}}; }
\WithoutColor{ \node [scale=1.1] at (32,17.5) {$M_{\Workers}$ and $M_{\Firms}$}; }

\node (x12k) at (31,16) [vertexWF] {$(1,2)$};

\node (x21k) at (29,14) [vertexWF] {$(2,1)$};

\node (x35k) at (35,12) [vertexWF] {$(3,5)$};

\node (x43k) at (33,10) [vertexWF] {$(4,3)$};

\node (x54k) at (34,9) [vertexWF] {$(5,4)$};

\draw [line width=0.03cm] (36,19) -- (36,7) -- (20,7) -- (20,19) -- (36,19);
\draw [line width=0.03cm] (28,19) -- (28,7);

\draw [->,line width=0.08cm] (16,24) -- (20,19);
\draw [->,line width=0.08cm] (10,24) -- (6,19);

\draw [->,line width=0.08cm] (13,41) -- (6,36);
\draw [->,line width=0.08cm] (19,41) -- (24,36);


\node (xww) at (28,52) [vertexWW] {$(-,-)$};
\node (xff) at (28,49) [vertexFF] {$(-,-)$};
\node (xwf) at (28,46) [vertexWF] {$(-,-)$};
\node [scale=0.9] at (32.4,52) {= pairs in $M_{\Workers}$};
\node [scale=0.9] at (32.4,49) {= pairs in $M_{\Firms}$}; 
\node [scale=0.9] at (32.4,46) {= pairs in both}; 
\node [scale=0.9] at (33,44.8) {$M_{\Workers}$ and $M_{\Firms}$};


\end{tikzpicture}
\caption{The binary search-tree.}
\label{fig:SearchTree}
\end{figure}
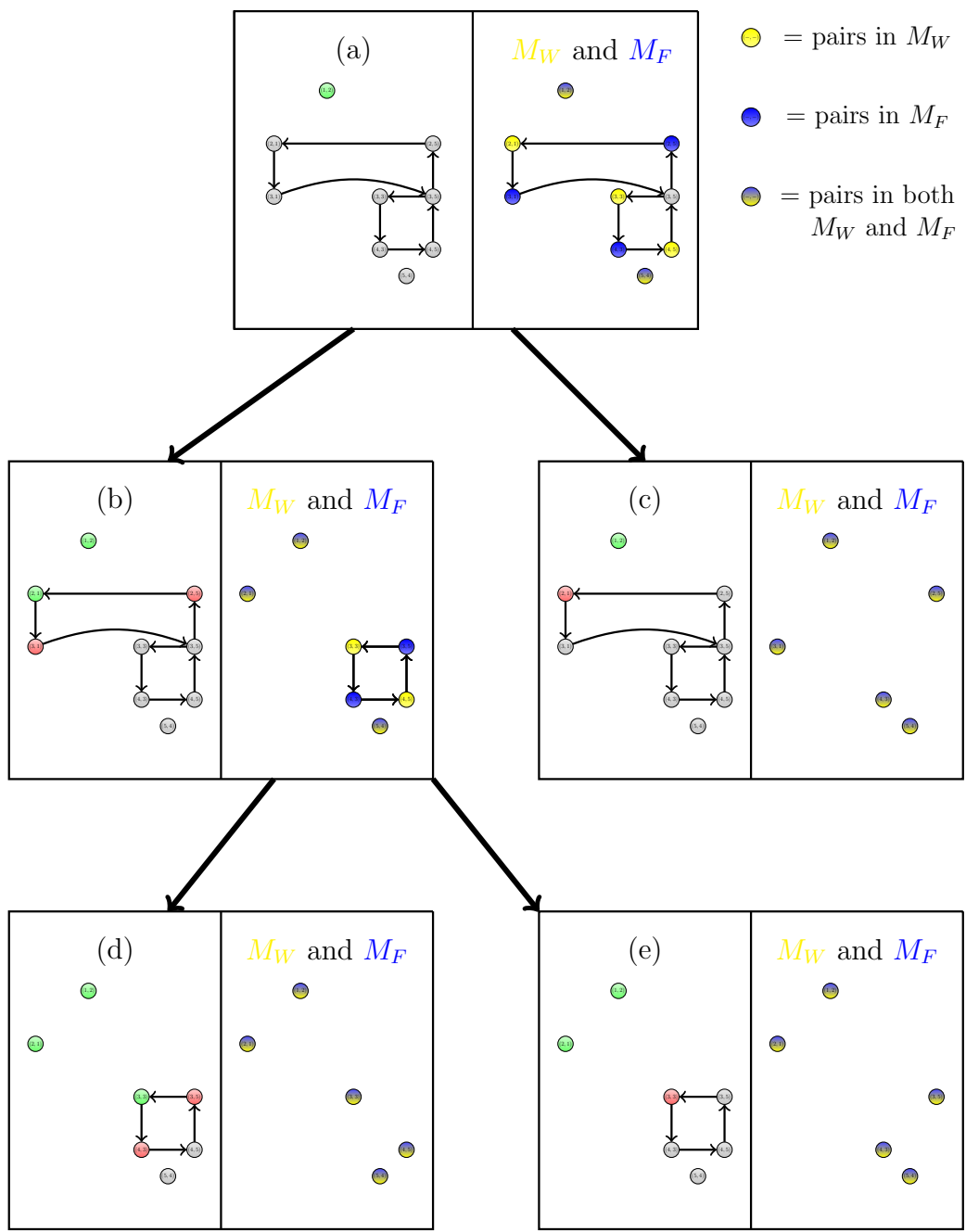

The first pair of panels in Figure~\ref{fig:SearchTree}, labelled (a), depicts the first time that Step~\ref{alg:ALLrec}.6 of Algorithm~\ref{alg:ALLrec} is called.
The left-hand panel reproduces panel (b) from Figure~\ref{fig:ex3} while the right hand panel reproduces Figure~\ref{fig:extremeConstrained}.
While there is nothing new depicted in these two panels, we include them to emphasise that this is the starting point of the recursive part of the algorithm.
From this point, the procedure ``splits'' indicated by the two thick black arrows from the top row to the second row.
One arrow goes ``left'' to the panels labelled (b) and the other goes ``right'' to the panels labelled (c).

Let us first consider the two adjoining panels in (b).
Note that in addition to vertex $(1,2)$ being required, it is now also the case that vertex $(2, 1)$ is required.
This is indicated by colouring vertex $(2,1)$ green in the left panel.
Given that vertex $(2,1)$ is required, it must be that vertices $(2, 5)$ and $(3, 1)$ cannot be simultaneously part of the stable matching (this is simply applying Step~\ref{alg:ALLrec}.1 to this subcase).
As such $(2, 5)$ and $(3, 1)$ are coloured red in the left hand panel.
Contrast this with the left panel of (c).
Here we note that vertex is $(2,1)$ is coloured red, indicating that it is forbidden.
In other words, (b) begins the part of the search tree where $(1,2)$ is required to be included and $(2,1)$ is required to be included; while (c) begins the part of the search tree where $(1,2)$ is required to be included and $(2,1)$ is required not to be included.
Clearly there can be no overlap in the set of stable matchings found by each branch.

In the right hand panel of (c), we see the consequences of requiring that vertex $(2,1)$ not be included.
Given this constraint, it must be that $(2,5)$ and $(3,1)$ are to be part of any stable matching.
Skipping some details, this immediately yields the stable matching $\set{(1,2), (2,5), (3,1), (4,3), (5,4)}$.

The same is not true of the right hand panel of (b).
Note that here there remains an anti-clockwise cycle involving vertices $\set{(3,3), (4,3), (4,5), (3,5)}$.
Again by appealing to Theorem 1 of \cite{GutinNeary:2023:GEB}, there must be at least two stable matchings satisfying the current constraints in this branch of the subtree.
Indeed this is true and is indicated by the fact that in the right hand panel of (b), $M_{\Workers{}} \neq M_{\Firms{}}$. Therefore, one more split is required.
The vertex that is pivoted around this time is $(3,3)$. In the left hand panel of (d), this vertex is coloured green indicating it must be included, while in the left hand panel of (e) this vertex is coloured red indicating it must not be included.

We conclude by noting that the algorithm identifies three stable matchings for the market of Reduced Example~\ref{exR:initial} that accord with the hypothetical market designer's assignment constraints as stipulated by {Reduced Question} \ref{q:QuestionR}. In order of decreasing worker preference, these three stable matchings are:
\begin{align*}
M^*_1 &= \set{(1,2), (2,1), (3,3), (4,5), (5,4)}		\\
M^*_2 &= \set{(1,2), (2,1), (3,5), (4,3), (5,4)}		\\
M^*_3 &= \set{(1,2), (2,5), (3,1), (4,3), (5,4)}
\end{align*}

As such, the market designer's Question \ref{q:QuestionGeneral} concerning the many-to-one matching market of Example~\ref{ex:initial} is answered in the affirmative with the three stable matchings satisfying the assignment constraints given by:
\begin{align*}
\match^*_1 &= \set{(\worker{1},\firm{2}), (\worker{2},\firm{1}), (\worker{3},\firm{3}), (\worker{4},\firm{4}), (\worker{5},\firm{4})}		\\
\match^*_2 &= \set{(\worker{1},\firm{2}), (\worker{2},\firm{1}), (\worker{3},\firm{4}), (\worker{4},\firm{3}), (\worker{5},\firm{4})}		\\
\match^*_3 &= \set{(\worker{1},\firm{2}), (\worker{2},\firm{4}), (\worker{3},\firm{1}), (\worker{4},\firm{3}), (\worker{5},\firm{4})}
\end{align*}


\newpage

\clearpage

\appendix

\section*{APPENDIX}\label{APP}

\section{Digraph terminology and notation}\label{sec:notation}
A {\em directed graph} (or just {\em digraph}) $D$ consists of a non-empty finite set $V (D)$ of elements called {\em vertices} and a finite set $A(D)$ of ordered pairs of distinct vertices called {\em arcs}. We shall call $V(D)$ the {\em vertex set} and $A(D)$ the {\em arc set} of $D$ and write $D=(V(D),A(D)).$ For an arc $xy$ the first vertex $x$ is its {\em tail} and the second vertex $y$ is its {\em head}. Moreover, $x$ is called an {\em in-neighbour} of $y$ and $y$ an {\em out-neighbour} of $x.$ We also say that the arc $xy$ {\em leaves} $x$ and {\em enters} $y$. 
We say that a vertex $x$ is {\em incident} to an arc $a$ if $x$ is the head or tail of $a$. For a vertex $v \in V(D)$, the {\em out-degree} of $v$ in $D$, $d^{+}_{D}(v)$, is the number of out-neighbours of $v$
Similarly, the {\em in-degree} of $v$ in $D$, $d^{-}_{D}(v)$, is the number of in-neighbours of $v$.
A vertex $u$ is {\em isolated} if $d^{+}_{D}(u)=d^{-}_{D}(u)=0.$

{
A {\em walk}, $W$, in a digraph $D$ is a sequence of vertices $x_1,x_2,\dots , x_p$ for which there is an arc from each vertex in the sequence to its successor in the sequence.  Such a walk is written as $W = x_1x_2\dots x_p$. Special cases of walks are paths and cycles. A walk $W$ is a {\em path} if the vertices of $W$ are distinct. If the vertices $x_1, x_2, \dots, x_{p-1}$ are distinct, for $p \geq 2$ and $x_1 = x_k$, then $W$ is a {\em cycle}.
}

For a digraph $D=(V,A)$ and an arc $xy\in A$, {\em deletion of $xy$ from} $D$ results in the digraph $D-xy=(V,A\setminus \{xy\}).$ For a vertex $v\in V$, {\em deletion of $v$ from} $D$ results in the digraph $D-v=(V\setminus \{v\},A\setminus A_v),$ where $A_v$ is the set of arcs in $A$ incident to $v$. A digraph $D'$ is called a {\em subdigraph} of $D$ if $D'$ is obtained from $D$ by deleting some vertices and arcs. 
If only vertices are deleted, $D'$ is an {\em induced subdigraph} of $D.$

\section{Proofs omitted from the main text}\label{app:proofs}

\begin{customlemma}{3}
In $D^*(P)$ there remain $r$ non-empty rows and $r$ non-empty columns (for some $r$) and all stable matchings in $D^*(P)$ contain $r$ vertices (that is, they are perfect matchings in $D^*(P)$). 
\end{customlemma}

\begin{proof}
Let $r_{\Workers}$ denote the number of non-empty rows in $D^*(P)$ and let 
$r_{\Firms}$ denote the number of non-empty columns in $D^*(P)$.
By definition, every row contains exactly one sink of that row (i.e., a vertex $v$ such that $d_{\Workers,i}^+(v)= 0$).
Moreover, by Lemma~\ref{lemma:extremalMatches} this vertex $v$ must be a source of a column (i.e., $d_{\Firms,i}^-(v)= 0$) and each column contains at most one source.
Putting this together, we have that $r_{\Firms} \geq r_{\Workers}$.
A mirror image argument shows that $r_{\Workers} \geq r_{\Firms}$, and so we have that $r_{\Workers} = r_{\Firms} = r$.

To show that all stable matchings in $D^*(P)$ are perfect matchings in $D^*(P)$ of size $r$, assume to the contrary. 
Then there must be a stable matching, $M$, of size strictly less than $r$ (clearly $M$ cannot contain more than $r$ pairs).
Now, let $\worker{}$ be an unmatched worker in $M$ and let $\firm{}$ denote the firm that has $\worker{}$ at the top of its preference list in $D^*(P)$ (i.e., the vertex $(\worker{}, \firm{})$ is the sink of the column corresponding to firm $\firm{}$).
Either (i) $\firm{}$ is not matched in $M$, or (ii) $\firm{}$ is paired with worker $\worker{}' \neq \worker{}$ in $M$.
No matter which holds, the matching $M$ is not stable as the pair $(\worker{}, \firm{})$ is a strict improvement for both $\worker{}$ and $\firm{}$.
\end{proof}

\2

\begin{customlemma}{4}
Let $D_i$ be any digraph obtained by running IDUA on some instance of the one-to-one two-sided matching market.
Assume that $D_i$ has $r$ rows and $r$ columns (and therefore, by Lemma~\ref{lemma:lem1}, all stable matchings are perfect matchings of size $r$).
Let $v \in V(D_i)$ be arbitrary, such that $d_{\Workers,i}^+(v)=0$ {or $d_{\Firms,i}^+(v)=0$} in $D_i$.
Then, the following two conditions are equivalent.
\begin{enumerate}[label=(\alph*)]
\item\label{eq:notInA}
$M$ is a stable matching in $D_i$ not containing $v$.

\item\label{eq:notInB}
$M$ is a stable matching in $D_i - v$.
\end{enumerate}
\end{customlemma}

\begin{proof}
Without loss of generality assume that  $d_{\Workers,i}^+(v)=0$.

We first show that \ref{eq:notInB} implies \ref{eq:notInA}.
Assume that $M$ is a stable matching in $D_i-v$ of size $r$.
Clearly $v \not\in M$ since $v \not\in D_i-v$, and so it suffices to show that $M$ is a stable matching in $D_i$.
As $d_{\Workers,i}^{+}(v)=0$ in $D_i$, we note that $v$ can have no arcs into it in $D_i$ from $A_{\Firms}$ 
(as $D_i$ is a digraph obtained by running IDUA on some instance).
But then, if $v=(\worker{},\firm{})$ and $u=(\worker{}',\firm{}) \in M$ ($u$ must exist as the size of $M$ is $r$), we must have that $vu \in A(D_i)$.
This implies that $v=(\worker{},\firm{})$ is not a blocking pair for $M$ in $D_i$.
Finally, because no vertex in $D_i-v$ is a blocking pair for $M$ in $D_i$ (as they are not blocking for $M$ in $D_i-v$), we have that $M$ must be a stable matching in $D_i$.

Now assume that \ref{eq:notInA} holds so that $M$ is a stable matching in $D_i$ not containing $v$.
Clearly $M$ is also a stable matching in $D_i-v$ because no vertex in $D_i$ is blocking for $M$, and so  \ref{eq:notInB} holds.

\end{proof}

\2

Before stating the theorem, recall that $m$ is the number of workers and $n$ is the number of firms in $P.$
Also, $N=\max\{m,n\}$, $M=\min\{m,n\}$ and $s$ is the number of stable matchings outputted by the algorithm.

\begin{customthm}{1}\label{}
Algorithm \ref{alg:ALLmain} correctly solves {\sc Reduced Constrained Stable Matchings} and 
and does so in time $O(N^3 + N^3s)$.  Moreover, the time between two consecutive stable matching outputted by the algorithm is $O(N^3)$.
\end{customthm}

By Lemma~\ref{lemma:iduaSameMatchings} the digraph $D^*(P)$ used in the call to Algorithm~\ref{alg:ALLrec} in Step~\ref{alg:ALLmain}.3 will contain 
exactly the same stable matchings as $D(P)$. So, if Algorithm~\ref{alg:ALLrec} adds all the desired stable matchings to $AllSolutions$ in the desired time, then Theorem~\ref{thm:ALLmain} holds.  We will now prove this.
The proof is partitioned into two parts, the proof of correctness and the proof of the running time claim. 
In what follows, stable matchings in $D(P)$ are referred to as {\em solutions}  and stable matchings that satisfy assignment constraints as {\em feasible solutions}. 

\vspace{2mm}

\noindent{\bf Proof of correctness.}

Clearly if $V^{in}\setminus V(D^{*})\ne \emptyset$ holds in Step~\ref{alg:ALLrec}.1, then there can be no solution in $D$ containing all the vertices in $V^{in}$, so the 
algorithm correctly returns without adding anything to $AllSolutions$.
Furthermore, as all solutions with $(w,f)$ in the solution contain no vertices  different from $(w,f)$, in the same row or column as $(w,f)$, we note that the 
second part of Step~\ref{alg:ALLrec}.1 does not change which solutions satisfy the constraints.

By Lemma~\ref{lemma:iduaSameMatchings} and Lemma~\ref{lem:NotIn} we note that every operation in Step~\ref{alg:ALLrec}.2 maintains the same set of feasible solutions. 

By Lemma~\ref{lemma:extremalMatches} we note that both $M_{\Workers}$ and $M_{\Firms}$ defined in Step~\ref{alg:ALLrec}.3 are stable matchings in $D$. As 
$(w,f) \not\in V^{out}$ for all $(w,f) \in V(D)$ where $d_{\Workers}^+((w,f))=0$ or $d_{\Firms}^+((w,f))=0$ we furthermore note that 
neither $M_{\Workers}$ or $M_{\Firms}$ contain any vertices from $V^{out}$. If 
$|M_{\Workers}|=r$ (which is equivalent to $|M_{\Firms}|=r$ as $|M_{\Workers}|=|M_{\Firms}|$) then by the operation in Step~\ref{alg:ALLrec}.1 we note that all
vertices in $V^{in}$ belong to both $M_{\Workers}$ and $M_{\Firms}$ (as no vertices from the rows or columns containing vertices in $V^{in}$ 
belong to $M_{\Workers}$ or $M_{\Firms}$). So if $|M_{\Workers}|=r$ then both $M_{\Workers}$ and $M_{\Firms}$ are stable matchings satisfying the 
constraints. 

If $|M_{\Workers}|<r$ in Step~\ref{alg:ALLrec}.4, then clearly there is no stable matching in $D$ of size $r$, which implies that no stable matching can be added to $AllSolutions$ (as all solutions in $AllSolutions$ must have size $r$, by Step~1.2 and Lemma~\ref{lemma:lem1}). So if Algorithm~\ref{alg:ALLrec} returns in Step~\ref{alg:ALLrec}.4, then there is no feasible solution in $D$.

If Algorithm~\ref{alg:ALLrec} does not return at Step~\ref{alg:ALLrec}.4, we can be sure that the digraph $D$ being considered at Step~\ref{alg:ALLrec}.5 has both $|M_{\Workers}|=r$ and $|M_{\Firms}|=r$, where $r$ is the size of a stable matching. If $M_{\Workers}=M_{\Firms}$ in Step~\ref{alg:ALLrec}.5, then $M_{\Workers}$ is the unique solution in $D$. We know this to be the case by appealing to Lemma 3 and Theorem 1 in \cite{GutinNeary:2023:GEB}. That is, Lemma 3 in \cite{GutinNeary:2023:GEB} confirms that both $M_{\Workers}$ and $M_{\Firms}$ are always perfect stable matchings of size $r$ in the normal form, and Theorem 1 in \cite{GutinNeary:2023:GEB} shows that if $M_{\Workers}=M_{\Firms}$ then there is a unique stable matching of size $r$ in the normal form.

In Step~\ref{alg:ALLrec}.6 we note that every solution in $D$ that satisfies the constraints will be found in the recursive call in Step~\ref{alg:ALLrec}.6a if $(w,f)$ belongs to the solution and it will be found in the recursive call in Step~\ref{alg:ALLrec}.6b if $(w,f)$ does not belongs to the solution. 
So every solution will be found in one of the recursive calls. Furthermore we note that no solution can be found in both the recursive calls in Step~\ref{alg:ALLrec}.6 as the solution either contains $(w,f)$ or it does not.
Therefore all solutions will be found and no solution will be found more than once. 
This proves the correctness of the algorithm. \qed

\vspace{2mm}

Before we prove the time complexity claim, we state and prove two useful lemmas.

We note that all steps, that do not involve a call to Algorithm~\ref{alg:ALLrec}, in Algorithm~\ref{alg:ALLmain} and in Algorithm~\ref{alg:ALLrec} can be performed in polynomial time. This is not difficult to see as IDUA runs in polynomial time and in Step~\ref{alg:ALLrec}.2 we cannot delete more than a polynomial number of vertices. We will now show the following:

\begin{lemma}\label{lem:T}
The number of calls to Algorithm~\ref{alg:ALLrec} is bounded by $2s+1$, where $s$ is the total number of feasible solutions outputted by Algorithm~\ref{alg:ALLmain}. 
\end{lemma}

\begin{proof}
Consider the binary search tree $T$ where every node corresponds to a call to Algorithm~\ref{alg:ALLrec} and the root corresponds to the very first call to Algorithm~\ref{alg:ALLrec}. 
If, from an iteration of Algorithm~\ref{alg:ALLrec} corresponding to node $x$ in $T$, we make two recursive calls in Step~\ref{alg:ALLrec}.6 then the two nodes corresponding to these
two calls will be the two children of $x$ in the search tree.
And if the iteration returns in Steps~\ref{alg:ALLrec}.1, \ref{alg:ALLrec}.4 or \ref{alg:ALLrec}.5 then $x$ is a leaf of the tree.
Therefore every node of the tree either has two children or none (i.e., it is a leaf). Let $l$ denote the number of leaves in $T$. It is well-known that $|V(T)| \leq 2l-1$ (this is true if $|V(T)|=l=1$ and if $|V(T)|>1$ we can use induction on the two subtrees of the root). 

As the stable matching $M_{\Workers}$ will be found in the recursive call made in Step~\ref{alg:ALLrec}.6a and the
stable matching $M_{\Firms}$ will be found in the recursive call made in Step~\ref{alg:ALLrec}.6b we note that the only way that a call to Algorithm~\ref{alg:ALLrec} will not produce a solution is if it is made from Algorithm~\ref{alg:ALLmain}. Therefore every leaf, apart from possibly the root, will add a solution to $AllSolutions$. Therefore the following holds.
\[
|V(T)| \leq 2l-1 \leq 2(s+1)-1 = 2s+1
\]
As $|V(T)|$ is the number of calls to Algorithm~\ref{alg:ALLrec} and $s$ was the total number of feasible solutions returned, we note that our algorithm performs a polynomial amount of work per solution found (if there is no solution the algorithm also performs a polynomial amount of work).  This completes the proof.
\end{proof}

Recall that $N=\max\{n,m\}$.
Now, for each worker, $w$, let $f'(w)$ denote the firm such that $d_{\Workers}^+((w,f'(w)))=0$, and for each firm, $f$, let $w'(f)$ denote the worker such that $d_{\Firms}^+((w'(f),f))=0$.
The workers and firms above can be found in $O(nm)$ time (we need $O(m)$ time for each row in $V(D)$ and $O(n)$ time for each column in $V(D)$).

To every vertex $(\worker{}, \firm{})$, we assign a flag, $\theta(\worker{}, \firm{})$, that tells us if the vertex belongs to $V^{in}$ or $V^{out}$ or neither.
These flags can also be initialised in $O(nm)$ time.
Another lemma required to show the desired complexity is the following:

\begin{lemma}\label{CompX}
If IDUA removes $p$ vertices from $D$ then it runs in at most $O(N + Np)$ time, even if it has to update $f'(w)$ and $w'(f)$ for 
some $w$ and/or $f$.
\end{lemma}

\begin{proof}
IDUA repeatedly performs the following steps.
\begin{itemize}
\item If for some $w$ we have $d_{\Firms}^-((w,f'(w)))>0$, then
remove all vertices $(w^*,f'(w))$ such that $(w^*,f'(w)) (w,f'(w)) \in A_{\Firms}$ and update $f'(w^*)$ accordingly.
\item If for some $f$ we have $d_{\Workers}^-((w'(f),f))>0$, then
remove all vertices $(w'(f),f^*)$ such that $(w'(f),f^*) (w'(f),f) \in A_{\Workers}$, and update $w'(f^*)$ accordingly.
\end{itemize}
Clearly, if no vertex is removed by IDUA, then 
$$d_{\Firms}^-((w,f'(w)))=0 \mbox{ and } d_{\Workers}^-((w'(f),f))=0$$ for all $w$ and $f$, 
and we have spent $O(N)$ time.
As updating $f'(w^*)$ and $w'(f^*)$ above can be done in constant time, we note that every time IDUA removes a vertex we have not spent
more than $O(N)$ time since removing the previous vertex, which completes the proof.
\end{proof}

\vspace{2mm}

\noindent{\bf Proof of complexity.}
By Lemma~\ref{CompX} we note that Step~\ref{alg:ALLmain}.1 in Algorithm~\ref{alg:ALLmain} can be done in time $O(N^3)$ time  (in fact it can be done in $O(N^2)$ time, but we do not need that here).
Step~\ref{alg:ALLmain}.2 can be done in $O(N)$ time and Steps~\ref{alg:ALLmain}.3 and \ref{alg:ALLmain}.4 can be done in constant time. So Algorithm~\ref{alg:ALLmain} can be performed in $O(N^3)$ time plus the complexity of Algorithm~\ref{alg:ALLrec}, which we analyse now.

In Algorithm~\ref{alg:ALLrec}, Step~\ref{alg:ALLrec}.1 can be done in $O(N^2)$ time. We now analyse Step~\ref{alg:ALLrec}.2. By considering the vertices $(w,f'(w))$ and $(w'(f),f)$ for all
$w$ and $f$, using the flags $\theta(\worker{}, \firm{})$, and using Lemma~\ref{CompX}, we note that Step~2.2 can be performed in $O(N^3 + Np)$ time, where $p$ is the number of vertices that 
get removed in Step~\ref{alg:ALLrec}.2 and $p \leq N^2$. Therefore Step~\ref{alg:ALLrec}.2 can be performed in $O(N^3)$ time.
Steps~\ref{alg:ALLrec}.3, \ref{alg:ALLrec}.4 and \ref{alg:ALLrec}.5 can all be performed in $O(N)$ time. Step~\ref{alg:ALLrec}.6 can be performed in $O(N^2)$ time even if we make a new copy of our instance to be used in each of the recursive calls.

So the total time spent in one iteration of Algorithm~\ref{alg:ALLrec} is $O(N^3)$ and as, by Lemma \ref{lem:T}, the number of such iterations is proportional to the number of solutions found, we obtain the complexity $O(N^3 + N^3s)$.

The analyses above of the running time of Algorithm~\ref{alg:ALLmain} without the complexity of Algorithm~\ref{alg:ALLrec} and of one iteration of Algorithm~\ref{alg:ALLrec} imply that the time between outputting two consecutive stable matchings by Algorithm~\ref{alg:ALLmain} is   $O(N^3)$. \qed

\section{Example showing that one can not just delete all vertices in $V^{out}$}\label{app:exampleDelete}

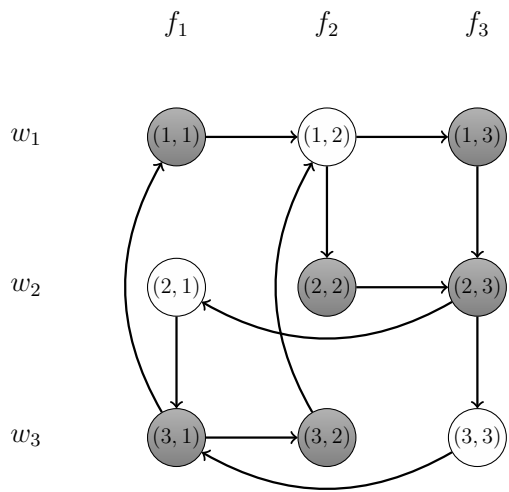
\begin{figure}[h!]
\begin{center}
\tikzstyle{vertexY}=[circle,draw, top color=gray!60, bottom color=gray!99, minimum size=16pt, scale=0.7, inner sep=0.99pt]
\tikzstyle{vertexZ}=[circle,draw, top color=gray!1, bottom color=gray!1, minimum size=16pt, scale=0.7, inner sep=0.99pt]
\begin{tikzpicture}[scale=1.0]
\node [scale=0.9] at (1,6.5) {$f_1$};
\node [scale=0.9] at (3,6.5) {$f_2$};
\node [scale=0.9] at (5,6.5) {$f_3$};

\node [scale=0.9] at (-1,5) {$w_1$};
\node (x11) at (1,5) [vertexY] {$(1,1)$};
\node (x12) at (3,5) [vertexZ] {$(1,2)$};
\node (x13) at (5,5) [vertexY] {$(1,3)$};

\node [scale=0.9] at (-1,3) {$w_2$};
\node (x21) at (1,3) [vertexZ] {$(2,1)$};
\node (x22) at (3,3) [vertexY] {$(2,2)$};
\node (x23) at (5,3) [vertexY] {$(2,3)$};

\node [scale=0.9] at (-1,1) {$w_3$};
\node (x31) at (1,1) [vertexY] {$(3,1)$};
\node (x32) at (3,1) [vertexY] {$(3,2)$};
\node (x33) at (5,1) [vertexZ] {$(3,3)$};

\draw [->,line width=0.03cm] (x11) -- (x12);
\draw [->,line width=0.03cm] (x12) -- (x13);
\draw [->,line width=0.03cm] (x13) -- (x23);
\draw [->,line width=0.03cm] (x23) -- (x33);
\draw [->,line width=0.03cm] (x33) to [out=210,in=330] (x31);
\draw [->,line width=0.03cm] (x31) -- (x32);
\draw [->,line width=0.03cm] (x32) to [out=120,in=240] (x12);
\draw [->,line width=0.03cm] (x12) -- (x22);
\draw [->,line width=0.03cm] (x22) -- (x23);
\draw [->,line width=0.03cm] (x23) to [out=210,in=330] (x21);
\draw [->,line width=0.03cm] (x21) -- (x31);
\draw [->,line width=0.03cm] (x31) to [out=120,in=240] (x11);

\end{tikzpicture}
\end{center}
\caption{The example depicted does not change by running IDUA on it since both extreme matchings are stable matchings. (Arcs implied by transitivity are not shown).}
\label{NewExV}
\end{figure}

Consider the matching problem $P$ equivalent to the digraph $D(P)$ in Figure~\ref{NewExV}.
Let $V^{out}=\{(1,1),(1,3),(2,2),(2,3),(3,1),(3,2)\}$
(the dark vertices in Figure~\ref{NewExV}). 
One cannot just remove all vertices in $V^{out}$ and then find a stable matching in the remaining digraph.
Doing so would find the matching $M=\{(1,2),(2,1),(3,3)\}$ (the white vertices) that is stable in the remaining subdigraph.
However $M$ is not a stable matching in the original problem as $(3,1)$ is a blocking pair for $M$.

In fact, if $V^{out}=\{(3,1)\}$ then we note that $M$ (the white vertices in Figure~\ref{NewExV}) is a stable matching in $D-(3,1)$, but $M$ is not a stable matching in $D$ avoiding $(3,1)$.
This explains why we can only delete vertices with no horizontal arcs out of them or vertices with no vertical arcs out of them in Step 2.2 of our algorithm.
This also illustrates why the condition $d_{\Workers,i}^+(v)=0$ (or $d_{\Firms,i}^+(v)=0$) is required in Lemma~\ref{lem:NotIn}.

\section{Example with exponential stable matchings}\label{app:exampleExponential}

It is always possible to solve {\sc Constrained Stable Matchings} by using one of the algorithms that outputs the full collection of stable matchings and then checking each against the assignment constraints.
One drawback to this approach is that the number of stable matchings can be exponential and so carrying out this procedure can be highly inefficient.
The inefficiency from proceeding in this manner is exacerbated if the number of stable matchings satisfying the constraints is small (or zero).

We highlight this potential issue with the following example.
Suppose that the normal form of the matching market under question has the digraph depicted in Figure~\ref{fig:multiple} with $r$ workers and $r$ positions.

\begin{figure}[h!]
\begin{center}
\tikzstyle{vertexDOT}=[scale=0.23,circle,draw,fill]

\tikzstyle{vertexZ}=[circle,draw, top color=\matchingC{}!20, bottom color=\matchingC{}!50, minimum size=11pt, scale=0.7, inner sep=0.99pt]
\tikzstyle{vertexWF}=[circle,draw, top color=blue!70, bottom color=yellow!100, minimum size=11pt, scale=0.5, inner sep=0.99pt]
\tikzstyle{vertexF}=[circle,draw, top color=blue!50, bottom color=blue!50, minimum size=11pt, scale=0.7, inner sep=0.99pt]

\tikzstyle{vertexY}=[circle,draw, top color=gray!10, bottom color=gray!40, minimum size=11pt, scale=0.7, inner sep=0.99pt]

\WithColor{\tikzstyle{vertexIn}=[circle,draw, top color=green!20, bottom color=green!60, minimum size=11pt, scale=0.5, inner sep=0.99pt]}
\WithoutColor{\tikzstyle{vertexIn}=[circle,draw, top color=gray!1, bottom color=gray!1, minimum size=21pt, scale=0.5, inner sep=2.99pt]}

\WithColor{\tikzstyle{vertexOut}=[circle,draw, top color=red!20, bottom color=red!60, minimum size=11pt, scale=0.5, inner sep=0.99pt]}
\WithoutColor{\tikzstyle{vertexOut}=[circle,draw, top color=gray!99, bottom color=gray!99, minimum size=21pt, scale=0.5, inner sep=2.99pt]}

\begin{tikzpicture}[scale=0.8]
\node [scale=0.9] at (1,12.5) {$f_1$};
\node [scale=0.9] at (2,12.5) {$f_2$};
\node [scale=0.9] at (3,12.5) {$f_3$};
\node [scale=0.9] at (4,12.5) {$f_4$};
\node [scale=0.9] at (5.1,12.5) {\dots};
\node [scale=0.9] at (5.9,12.5) {\dots};

\node [scale=0.9] at (-1,11) {$w_1$};
\node (x11) at (1,11) [vertexZ] {$$};
\node (x12) at (2,11) [vertexF] {$$};

\node [scale=0.9] at (-1,10) {$w_2$};
\node (x21) at (1,10) [vertexF] {$$};
\node (x22) at (2,10) [vertexZ] {$$};

\node [scale=0.9] at (-1,9) {$w_3$};
\node (x33) at (3,9) [vertexZ] {$$};
\node (x34) at (4,9) [vertexF] {$$};

\node [scale=0.9] at (-1,8) {$w_4$};
\node (x43) at (3,8) [vertexF] {$$};
\node (x44) at (4,8) [vertexZ] {$$};

\node [scale=0.9] at (-1,7) {\vdots};

\node [scale=0.9] at (-1,6.5) {\vdots};

\node [scale=0.9] at (-1,5) {$w_{r-3}$};
\node [scale=0.9] at (-1,4) {$w_{r-2}$};
\node [scale=0.9] at (7,12.5) {$f_{r-3}$};
\node [scale=0.9] at (8,12.5) {$f_{r-2}$};
\node (x55) at (7,5) [vertexZ] {$$};
\node (x56) at (8,5) [vertexF] {$$};
\node (x65) at (7,4) [vertexF] {$$};
\node (x66) at (8,4) [vertexZ] {$$};
\draw [->,line width=0.03cm] (x55) -- (x65);
\draw [->,line width=0.03cm] (x65) -- (x66);
\draw [->,line width=0.03cm] (x66) -- (x56);
\draw [->,line width=0.03cm] (x56) -- (x55);

\node [scale=0.9] at (-1,3) {$w_{r-1}$};
\node [scale=0.9] at (-1,2) {$w_r$};
\node [scale=0.9] at (9,12.5) {$f_{r-1}$};
\node [scale=0.9] at (10,12.5) {$f_r$};
\node (x77) at (9,3) [vertexZ] {$$};
\node (x78) at (10,3) [vertexF] {$$};
\node (x87) at (9,2) [vertexF] {$$};
\node (x88) at (10,2) [vertexZ] {$$};
\draw [->,line width=0.03cm] (x77) -- (x87);
\draw [->,line width=0.03cm] (x87) -- (x88);
\draw [->,line width=0.03cm] (x88) -- (x78);
\draw [->,line width=0.03cm] (x78) -- (x77);

\draw [->,line width=0.03cm] (x12) -- (x11);

\draw [loosely dotted,line width=0.03cm] (4.4,7.6) -- (6.6,5.4);

\draw [->,line width=0.03cm] (x21) -- (x22);

\draw [->,line width=0.03cm] (x34) -- (x33);

\draw [->,line width=0.03cm] (x43) -- (x44);



\draw [->,line width=0.03cm] (x11) -- (x21);

\draw [->,line width=0.03cm] (x22) -- (x12);

\draw [->,line width=0.03cm] (x33) -- (x43);

\draw [->,line width=0.03cm] (x44) -- (x34);


\end{tikzpicture}
\end{center}
\caption{Market with $2^\frac{n}{2}$ stable matchings.}

\label{fig:multiple}
\end{figure}
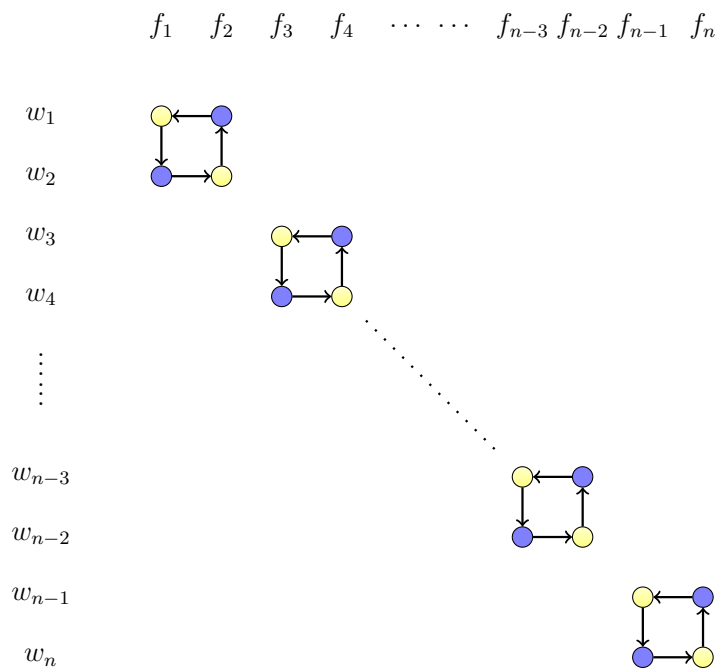

Each 4-cycle contains two stable components, one worker optimal and one firm optimal.
So the market as a whole has $2^\frac{r}{2}$ stable matchings.
However, while there are exponential number of stable matchings, we will see that only four of them satisfy the assignment constraints.
We have highlighted the two extremal stable matchings using colour coding as before (worker-optimal in {yellow} and firm-optimal in {blue}).

Consider constraints of the form $V^{out}=\set{(k,k)}_{k=5}^{r}$.
That is, for all $k \geq 5$, the constraints do not permit that a worker with index $k$ be employed at a position with the same index.
Referencing Figure~\ref{fig:multiple}, these constraints rule out any stable matching that includes a pair shaded yellow, other than those yellow vertices in positions $(1,1), (2,2), (3,3)$ and $(4,4)$.
As such, from index $k=5$ onwards, only blue vertices can be selected.
We note that there are only four stable matchings satisfying the constraints.

This example is one that, to the best of our knowledge, only our algorithm can handle efficiently.
There are an exponential number of stable matchings but only a small number, specifically four, satisfying the constraints.
(While we concede that the example is somewhat artificial, we have chosen it as other, potentially more realistic, examples would not be so cleanly depicted.)

\newpage

\bibliographystyle{plainnat}
\bibliography{withWithoutReferences.bib}

\end{document}